\documentclass[a4paper,UKenglish]{lipics}
\pdfoutput=1

\usepackage{microtype}

\newif\ifignore 
\ignorefalse    

\newcommand{\auxproof}[1]{\ifignore
\mbox{}\newline
\textbf{PROOF:} \dotfill\newline
{\small#1}\mbox{}\newline
\textbf{ENDPROOF}\dotfill\newline
\fi}

\usepackage{hyperref}

\usepackage{latexsym}
\usepackage{amsmath}
\usepackage{amssymb}
\usepackage{amsthm}
\usepackage{mathtools}
\usepackage{mathrsfs}
\usepackage{stmaryrd}
\usepackage{nicefrac}
\usepackage[all]{xy}
\SelectTips{cm}{}
\newdir{ (}{{}*!/-5pt/@{(}}
\newdir^{ (}{{}*!/-5pt/@^{(}}
\newdir_{ (}{{}*!/-5pt/@_{(}}

\usepackage[noadjust]{cite}
\usepackage{bussproofs}
\usepackage{xparse}
\usepackage{enumitem}
\usepackage{todonotes}
\presetkeys{todonotes}{size=\footnotesize,color=orange!30}{}
\usepackage{wrapfig}

\theoremstyle{plain}
\newtheorem{proposition}[theorem]{Proposition}

\theoremstyle{definition}
\newtheorem{nremark}[theorem]{Remark} 
\newtheorem{notation}[theorem]{Notation}

\newcounter{enumTempi}

\newcommand{\nsp}{\mathrm{nsp}}

\usepackage{comment}
\ifignore
\newenvironment{Auxproof}
  {\par\noindent\textbf{BEGIN: AUX-PROOF}\dotfill\par\footnotesize}
  {\normalsize\par\noindent\textbf{END: AUX-PROOF}\dotfill\par}
\else
\excludecomment{Auxproof}
\fi

\providecommand{\phantomsection}{}
\AtBeginDocument{\let\textlabel\label}
\makeatletter
\newcommand{\mylabel}[2]{\raisebox{.7\normalbaselineskip}{\phantomsection}#1%
  \def\@currentlabel{#1}\textlabel{#2}}

\newcommand{\MyRightLabel}[2]{\RightLabel{(\mylabel{#1}{rule@#2})}}
\newcommand{\rulelabel}[1]{\label{rule@#1}}
\newcommand{\ruleref}[1]{(\ref{rule@#1})}
\makeatother

\renewcommand{\emptyset}{\varnothing}
\newcommand{\N}{\mathbb{N}}

\newcommand{\C}{\mathbb{C}}
\DeclarePairedDelimiter{\abs}{\lvert}{\rvert}

\newcommand{\linf}{\ell^{\infty}}
\DeclarePairedDelimiter{\sem}{\llbracket}{\rrbracket}
\DeclarePairedDelimiter{\bra}{\langle}{\rvert}
\DeclarePairedDelimiter{\ket}{\lvert}{\rangle}

\DeclarePairedDelimiter{\tup}{\langle}{\rangle}
\newcommand{\Mat}{\mathcal{M}}
\DeclarePairedDelimiterX{\innp}[2]{\langle}{\rangle}%
{#1\,\delimsize\vert\,\mathopen{}#2}
\DeclarePairedDelimiterX{\braket}[3]{\langle}{\rangle}%
{#1\delimsize\vert\mathopen{}#2\delimsize\vert\mathopen{}#3}
\DeclarePairedDelimiterX{\inprod}[2]{\langle}{\rangle}%
{#1\delimsize\vert\mathopen{}#2}

\newcommand{\To}{\Rightarrow}
\newcommand{\longto}{\longrightarrow}
\newcommand{\id}{\mathrm{id}}
\newcommand{\op}{\mathrm{op}}
\newcommand{\bang}{\mathord{!}}

\newcommand{\invbangsub}[2]{\rotatebox[origin=c]{180}{$#1#2$}}
\newcommand{\invbang}{\mathord{\mathpalette\invbangsub{!}}}

\newcommand{\vNA}{\mathbf{vNA}}

\newcommand{\vNAMIU}{\vNA_{\mathrm{MIU}}}

\newcommand{\vNACPsU}{\vNA_{\mathrm{CPsU}}}
\newcommand{\Set}{\mathbf{Set}}

\newcommand{\qcomp}{\mathcal{F}}
\newcommand{\incfun}{\mathcal{J}}
\newcommand{\frexp}[2]{{#2}^{*#1}}
\newcommand{\lem}{\mathcal{L}}

\newcommand{\scrA}{\mathscr{A}}
\newcommand{\scrB}{\mathscr{B}}
\newcommand{\scrC}{\mathscr{C}}

\newcommand{\monmap}{\triangledown}

\newcommand{\prob}{\mathop{\mathrm{prob}}\nolimits}
\newcommand{\Prob}{\mathop{\mathrm{Prob}}\nolimits}
\newcommand{\mergmap}{\mathrm{merge}}
\newcommand{\splmap}{\mergmap}

\DeclareMathOperator{\FV}{FV}
\newcommand{\oc}{\mathord{!}}
\newcommand{\unittype}{\top}
\newcommand{\unitterm}{\mathord{*}}
\newcommand{\subtype}{\mathrel{<\vcentcolon}}

\newcommand{\limp}{\mathbin{\multimap}}
\newcommand{\qclos}[3]{[\,#1,#2,#3\,]}
\newcommand{\tlam}[3][]{\lambda^{#1}#2.#3}
\newcommand{\ttup}[1]{\langle#1\rangle}
\DeclareDocumentCommand{\tinl}{ O{} O{} m }
{\mathop{\mathtt{inl}}\nolimits^{#1}_{#2}(#3)}
\DeclareDocumentCommand{\tinr}{ O{} O{} m }
{\mathop{\mathtt{inr}}\nolimits^{#1}_{#2}(#3)}
\newcommand{\tletin}[3]{\mathtt{let}\;{#1}={#2}\;\mathtt{in}\;{#3}}
\newcommand{\tmatchwith}[6][]{\mathtt{match}\;{#2}\;\mathtt{with}
\if\relax\detokenize{#1}\relax
\;\else^{#1}\fi
({#3}\;{\mapsto}\;{#4}
\;|\;{#5}\;{\mapsto}\;{#6})}
\newcommand{\tifthen}[3]{\mathtt{if}\;{#1}\;\mathtt{then}
\;{#2}\;\mathtt{else}\;{#3}}
\newcommand{\yields}{\mathrel{\vartriangleright}}
\newcommand{\tbit}{\mathtt{bit}}
\newcommand{\tzero}{\mathtt{f\!f}}
\newcommand{\tone}{\mathtt{t\!t}}
\newcommand{\tqbit}{\mathtt{qbit}}
\newcommand{\tnew}{\mathop{\mathtt{new}}\nolimits}
\newcommand{\tmeas}{\mathop{\mathtt{meas}}\nolimits}
\newcommand{\tsubst}[3]{#1[#3/#2]}
\newcommand{\tsubsttwo}[5]{#1[#3/#2,#5/#4]}

\newcommand{\redu}[1][]{\to_{#1}}
\newcommand{\bigredu}[2]{#1\Downarrow#2}
\newcommand{\permop}[1]{\bar{#1}}
\newcommand{\Iop}{\mathcal{I}}

\newcommand{\tfst}[2][]{\mathop{\mathtt{fst}}\nolimits_{#1}(#2)}
\newcommand{\tsnd}[2][]{\mathop{\mathtt{snd}}\nolimits_{#1}(#2)}
\newcommand{\transl}[1]{#1^{\dagger}}

\title{Von Neumann Algebras form a
   Model for the Quantum Lambda Calculus\footnote{%
The research leading to these results has received funding from the
European Research Council under the European Union's Seventh Framework
Programme (FP7/2007-2013) / ERC grant agreement n\textsuperscript{o} 320571}}

\author{Kenta Cho}
\author{Abraham Westerbaan}
\affil{Institute for Computing and Information Sciences\\
  Radboud University, Nijmegen, the Netherlands\\
  \texttt{\{K.Cho,awesterb\}@cs.ru.nl}}
\authorrunning{K. Cho and A. Westerbaan} 

\Copyright{Kenta Cho and Abraham Westerbaan}

\subjclass{F.3.2 Semantics of Programming Language}
\keywords{quantum lambda calculus, von Neumann algebras}

\begin{document}

\maketitle

\begin{abstract}
We present a model of Selinger and Valiron's quantum lambda calculus
based on von Neumann algebras,
and show that the model is adequate
with respect to the operational semantics.
\end{abstract}

\section{Introduction}

In 1925,
Heisenberg realised,
pondering upon the problem of the spectral lines
of the hydrogen atom,
that a physical quantity
such as the $x$-position of an electron
orbiting a proton
is best described
not by a real number
but by an infinite array
of complex numbers~\cite{heisenberg1925}.
Soon afterwards, Born and Jordan noted that these arrays
should be multiplied as matrices are~\cite{born1925}.
Of course,
multiplying
such infinite matrices
may  lead to mathematically dubious situations,
which spurred von Neumann
to replace
the infinite matrices by operators on a Hilbert space~\cite{neumann1927}.
He organised
these into \emph{rings of operators}~\cite{murray1936},
now called \emph{von Neumann algebras},
and thereby set off an explosion
of research 
(also into related structures such as Jordan algebras~\cite{jordan1933},
orthomodular lattices~\cite{birkhoff1936},
$C^*$-algebras~\cite{segal1947},
$AW^*$-algebras~\cite{kaplansky1951},
order unit spaces~\cite{kadison1951},
Hilbert $C^*$-modules~\cite{paschke1973},
operator spaces~\cite{ruan1988},
effect algebras~\cite{foulis1994},
\dots), which continues even to this day.

One current line of research
(with old roots~\cite{guichardet1966,dauns1972,lorenz1969,dauns1978})
is the study of
von Neumann algebras
from a categorical 
perspective (see e.g.~\cite{cho2014,rennela2014,cho2015}).
One example
relevant to this paper
is Kornell's proof
that the opposite of the category~$\vNAMIU$
of von Neumann algebras
with the obvious structure preserving maps
(i.e.~the unital normal $*$-homomorphisms)
is monoidal closed
when 
endowed with the spatial tensor product~\cite{Kornell2012}.
He argues
that~$\vNAMIU^\op$
should 
be thought of as the quantum
version of~$\Set$.
We would like 
to focus instead on the category of von Neumann algebras
and completely positive normal subunital
maps, $\vNACPsU$,
as it seems more appropriate for modelling quantum computation:
the full subcategory of~$\vNACPsU^\op$ consisting of finite dimensional 
von Neumann algebras
is equivalent to Selinger's category~$\mathbf{Q}$~\cite{selinger2004},
which is used
to model first order quantum programming languages.

\begin{wrapfigure}{r}{.3\columnwidth}
\vspace{-.5em}
\centering
$
\xymatrix{
\mathscr{B} 
\ar@/^1em/[r]
\ar@{}[r]|\bot
&
\mathscr{C} 
\ar@/^1em/[r]
\ar@{}[r]|\bot
\ar@/^1em/[l]
&
\mathscr{D}
\ar@/^1em/[l]
}
$
\caption{General shape of a  model
of the QLC}
\label{fig:model}
\vspace{-1.5em}
\end{wrapfigure}
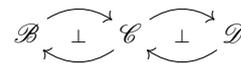
On the syntactic side,
in 2005, Selinger and Valiron~\cite{SelingerV2005,SelingerV2006} proposed a typed%
\footnote{%
An untyped quantum lambda calculus
had already been proposed by Van Tonder~\cite{vantonder2004}.}
lambda calculus for quantum computation,
and they studied it in a series of
papers~\cite{SelingerV2008a,SelingerV2008,SelingerV2009}.
A striking feature
of this quantum lambda calculus
is that functions naturally appear as data in the description
of the Deutch--Jozsa algorithm,
teleportation algorithm
and Bell's experiment.
Although 
Selinger and Valiron
gave a precise formulation
of what might constitute a model of the quantum lambda calculus
---
basically a pair
of adjunctions,
see Figure~\ref{fig:model},
with some additional properties~\cite[\S1.6]{SelingerV2009} ---
the existence of such a model
(other than the term model)
was an open problem for several years until Malherbe
constructed a model in his thesis
using presheaves~\cite{Malherbe2010}.
The construction
of Malherbe's
model is quite abstract,
and it is (perhaps because of this)
not yet known whether his model is adequate
with respect to the operational semantics
defined by Selinger and Valiron in \cite{SelingerV2006}
(see also \cite{SelingerV2009}).
While several
 adequate
models for variations on the quantum lambda calculus
have been proposed
in the meantime
(using the geometry of interaction in~\cite{HasuoH2011},
and quantitative semantics in~\cite{PaganiSV2014}),
Malherbe's model
remains the only model of the 
original quantum lambda calculus~\cite{SelingerV2006} known in the literature,
and so the existence of an adequate model
for the quantum lambda calculus is still open.

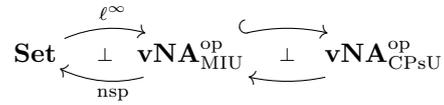
\begin{wrapfigure}{r}{.45\columnwidth}
\vspace{-1em}
\centering
$
\xymatrix{
\mathbf{Set}
\ar@/^1em/[r]^{\ell^\infty}
\ar@{}[r]|(.45){\bot}
&
\vNAMIU^\mathrm{op}
\ar@{^(->}@/^1em/[r]
\ar@{}[r]|\bot
\ar@/^1em/[l]^{\nsp}
&
\vNACPsU^\mathrm{op}
\ar@/^1em/[l]
}
$
\caption{A  model of the QLC
using von Neumann algebras}
\label{fig:neumann-model}
\vspace{-.5em}
\end{wrapfigure}
In this paper,
we present the model
 of Selinger and Valiron's quantum lambda calculus,
based on von Neumann algebras,
see Figure~\ref{fig:neumann-model},
and
we show that the model is adequate
with respect to the operational semantics.
We should note that it is possible
to extend the quantum lambda calculus
with recursion and inductive types,
but that we have not yet been able
to include these features in our model.

The paper is divided in six sections.
We begin with a short
review of  quantum computation
(in Section~\ref{sec:quantum-computation}),
and  the quantum lambda calculus and its 
operational semantics
(in Section~\ref{sec:quantum-lambda-calculus-op-sem}).
We give the
denotational semantics
for the quantum lambda calculus 
using von Neumann algebras
and prove its adequacy
in Section~\ref{sec:denotational-semantics}.
For this
we use several technical results
about the categories~$\vNAMIU$ and~$\vNACPsU$
of von Neumann algebras,
which we will discuss in Section~\ref{technical-result-vNA}.
We end with a conclusion in Section~\ref{sec:final-remarks}.

\section{Quantum Computation}
\label{sec:quantum-computation}

In a nutshell,
one gets
the quantum lambda calculus
by taking the simply typed lambda calculus
with products and coproducts
and adding a qubit type.
This single ingredient
dramatically changes the flavour of the whole system
e.g.~forcing one to make the type system linear,
so we will spend some words on the behaviour of qubits
in this section.
For more details on quantum computation, see~\cite{nielsen2010}.

A state of an isolated qubit
is a vector~$\ket\psi$ of length~$1$ 
in the Hilbert space~$\mathbb{C}^2$,
and can be written
as
a complex linear combination
(``superposition'')
$\ket\psi = \alpha \ket0 + \beta \ket1$,
since 
the vectors  $\ket0 = (1,0)$
and $\ket1 = (0,1)$
form an orthonormal basis
for~$\mathbb{C}^2$.

When qubits are combined to form a larger system,
one can sometimes no longer speak about the state of the individual qubits,
but only of the state of the whole system 
(in which case the qubits are ``entangled'').
The  state of a register of~$n$ qubits
is a vector~$\ket\psi$ of length~$1$ in
the $n$-fold tensor product 
$(\mathbb{C}^2)^{\otimes n}\cong \smash{\mathbb{C}^{2^n}}$,
which has 
as an orthonormal basis
the vectors
of the form
$\ket{w}\equiv \ket{w_1}\otimes \dotsb\otimes \ket{w_n}$
where $w\equiv w_1\dotsb w_n \in 2^n$.

For the purposes of this paper
there are three basic operations 
 on registers of qubits.
\begin{enumerate}
\item
One can add a new qubit
in state~$\ket0$
to a register of~$n$ qubits in state~$\ket\psi$,
turning it to a register of~$n+1$ qubits
in state $\ket\psi \otimes \ket 0$.
A qubit in state~$\ket1$
can be added similarly.

\item
One can apply
a unitary $2^n\times 2^n$ matrix~$U$
to a register of~$n$ qubits
in state~$\ket\psi$
turning the state to~$U\ket\psi$.

\item
One can test
the first qubit
in the register.
If the state
of the register is written as  $\ket\psi 
\equiv \alpha\, \ket 0 \otimes \ket{\psi_0}
\,+\, \beta\, \ket1 \otimes \ket{\psi_1}$
where the length of~$\ket{\psi_0}$
and~$\ket{\psi_1}$ is~$1$,
then the test comes out negative
and
changes
the state of the register to~$\ket0 \otimes \ket{\psi_0}$
with probability~$|\alpha|^2$,
and comes out positive
with probability~$|\beta|^2$
changing the state to~$\ket1 \otimes \ket{\psi_1}$.

Measurement of the $i$-th qubit in the register
is also possible and behaves similarly.
\end{enumerate}
A predicate on a register of~$n$ qubits
is a $2^n\times 2^n$ matrix~$P$
such that both $P$ and~$I-P$ are positive
(which is the case when~$P$ is a projection).
The probability that~$P$ holds in
state~$\ket\psi$ is $\braket{\psi}{P}{\psi}$.
For example,
given a state~$\ket\psi$ of a qubit,
the projection $\ket{\psi}\bra{\psi}$
(which maps $\ket\xi$ to~$\inprod{\psi}{\xi}\ket{\psi}$)
represents the predicate ``the qubit is in state~$\ket\psi$''.

Thus the predicates
on a qubit are part of the algebra~$\Mat_2$
of $2\times 2$ complex matrices.
There is also an algebra 
for the bit,
namely~$\mathbb{C}^2$.
A predicate on a bit
is an element~$(x,y)\equiv v\in \mathbb{C}^2$ with $0\leq v \leq 1$,
which is interpreted as
``the bit is true with probability~$y$,
false with probability~$x$,
and undefined with probability $1-x-y$''.

An operation on a register of qubits
may not only be described by the effect it has
on states (Schr\"odinger's view),
but also by its action on predicates (Heisenberg's view).
\begin{enumerate}
\item
The operation
which takes a bit~$b$
and returns a qubit in state $\ket{b}$
is represented by the map $f_{\tnew}\colon \Mat_2\to \mathbb{C}^2$
given by~$f_{\tnew}(A) = (\,\braket{0}{A}{0},\,\braket{1}{A}{1}\,)$.

\item
The operation
which applies a unitary~$U$
to a register of~$n$ qubits
is represented by the map~$f_{U}\colon \Mat_{2^n}\to\Mat_{2^n}$
given by~$f_U(A) = U^*AU$.

\item
The operation
which tests a qubit
and returns the outcome
is represented by the map~$f_{\tmeas}\colon \mathbb{C}^2\to\Mat_2$
given by~$f_{\tmeas}(\lambda,\varrho)= \lambda\ket{0}\bra{0}
+\varrho \ket{1}\bra{1}$.
\end{enumerate}

A general operation between \emph{finite dimensional} quantum data
types
is usually taken 
to be a completely positive subunital linear map (see below)
between direct sums
of matrix algebras, $\bigoplus_{i=1}^n \Mat_{m_i}$.
The category formed by these operations
is equivalent to~$\mathbf{Q}^\op$~\cite[Th.~8.4]{cho2014}.

Von Neumann algebras
are a generalisation
of direct sums of matrix algebras
to infinite dimensions.
Formally,
a von Neumann 
algebra~$\mathscr{A}$
is a linear subspace
of the bounded operators 
on a Hilbert space~$\mathscr{H}$,
which contains the identity operator, $1$, 
is closed under multiplication,
involution, $(-)^*$,
and is closed
in the weak operator topology,
i.e.~the topology
generated by the seminorms $|\braket{x}{-}{x}|$
where~$x\in\mathscr{H}$
(cf.~\cite{kadison1997,murray1936}).

We believe that
the opposite
$\vNACPsU^\op$
of the category 
of von Neumann algebras
and normal completely positive subunital maps
(definitions are given below)
might 
turn out to be the most suitable
extension
of~$\mathbf{Q}$
to describe
operations
between
(possibly infinite dimensional)
quantum data types.
Indeed,
to support this thesis,
we will show that~$\vNACPsU^\op$
gives a model of the quantum lambda calculus.

Let us end this section
with the definitions
that are necessary to understand~$\vNACPsU$.
An element~$a$ of
a von Neumann algebra~$\mathscr{A}$ is  \emph{self-adjoint}
if $a^*=a$,
and \emph{positive} if~$a\equiv b^*b$ for some~$b\in \mathscr{A}$.
The self-adjoint elements
of a von Neumann algebra~$\mathscr{A}$
are partially ordered by: $a\leq b$ iff $b-a$ is positive.
Any upwards directed bounded
subset~$D$ of self-adjoint 
elements of a von Neumann algebra~$\mathscr{A}$ has a supremum~$\bigvee D$
in the set of self-adjoint elements 
of~$\mathscr{A}$~\cite[Lem.~5.1.4]{kadison1997}.
(So a von Neumann algebra
resembles a domain.)

The linear maps between von Neumann algebras
which preserve the multiplication,
involution,~$(-)^*$, and unit,~$1$,
are called unital $*$-homomorphisms in the literature
and \emph{MIU-maps} by us.
A linear map~$f$ between von Neumann algebras
is \emph{positive}
if it maps positive elements to positive elements,
\emph{unital}
if it preserves the unit,
\emph{subunital}
if~$f(1)\leq 1$,
and \emph{normal}
if~$f$ is positive and 
preserves suprema of bounded directed sets of self-adjoint
elements. (If subunital maps
are akin to partial maps between sets,
then the unital maps are the total maps.
Normality is the incarnation
of Scott continuity
in this setting,
and coincides
with continuity with respect
to the $\sigma$-weak = ultraweak = weak* topology~\cite[Th.~1.13.2]{Sakai2012}.)

Given a von Neumann algebra~$\mathscr{A}$
on a Hilbert space~$\mathscr{H}$,
and a von Neumann algebra~$\mathscr{B}$
on a Hilbert space~$\mathscr{K}$,
the \emph{spatial tensor product}
$\mathscr{A}\otimes \mathscr{B}$
of~$\mathscr{A}$ and~$\mathscr{B}$
is the least von Neumann algebra on~$\mathscr{H}\otimes \mathscr{K}$
which contains all operators of the form~$a\otimes b$
where $(a\otimes b)(x\otimes y) = a(x)\otimes b(y)$
for all~$a\in\mathscr{A}$, $b\in \mathscr{B}$,
$x\in\mathscr{H}$ and~$y\in\mathscr{K}$~\cite[\S11.2]{KadisonRingrose2}.
(The tensor product~$\mathscr{A}\otimes\mathscr{B}$
may be physically interpreted as the composition
of the systems~$\mathscr{A}$ and~$\mathscr{B}$ --- recall
that a register of two
qubits is represented by the von Neumann algebra
$\Mat_2\otimes \Mat_2$.)

Given normal positive $f\colon\mathscr{A}\to\mathscr{B}$
and $g\colon \mathscr{C}\to\mathscr{D}$
there might
be a normal positive linear map
$f\otimes g\colon \mathscr{A} \otimes \mathscr{C}\to
\mathscr{B}\otimes\mathscr{D}$
given by $(f\otimes g)(a\otimes c)=f(a)\otimes g(c)$.
An interesting,
and annoying,
phenomenon
is that such $f\otimes g$ 
need not exist for all~$f$ and~$g$.
This warrants
the following definition:
if $f\colon\mathscr{A}\to\mathscr{B}$ 
is a positive linear map
such that for every natural number~$n$
the map
$\Mat_n(f)\colon \Mat_n(\mathscr{A})
\to\Mat_n(\mathscr{B})$
is  positive,
then~$f$ is called \emph{completely positive}~\cite{paulsen2002}.
Here~$\Mat_n(\mathscr{A})$
is the von Neumann algebra of~$n\times n$ matrices
with entries drawn from~$\mathscr{A}$,
and~$\Mat_n(f)(A)_{ij}=f(A_{ij})$
for all~$i,j$ and~$A\in\Mat_n(\mathscr{A})$.
If~$f$ and~$g$ are normal and completely positive,
then~$f\otimes g$ exists,
and is completely positive~\cite[Prop.~IV/5.13]{Takesaki1}.

\section{The Quantum Lambda Calculus and its Operational Semantics}
\label{sec:quantum-lambda-calculus-op-sem}

We review the quantum lambda calculus for which we will give
a denotational semantics. The language and its operational semantics
are basically the same as Selinger and Valiron's ones~\cite{SelingerV2006},
but with sum type $\oplus$~\cite{SelingerV2009}
and `indexed' terms~\cite{SelingerV2008},
see Remark~\ref{rem:indexedterms} and Notation~\ref{notation:sumtype} below.
For space reasons we omit many details,
and refer to \cite{SelingerV2006,SelingerV2008,SelingerV2009}.

\subsection{Syntax and Typing Rules}

\begin{table}[tb]
\small
\setlength{\abovedisplayskip}{0pt}
\setlength{\belowdisplayskip}{0pt}
\setlength{\abovedisplayshortskip}{0pt}
\setlength{\belowdisplayshortskip}{0pt}
\begin{gather*}
\text{\emph{Type}}\;\;
A,B
\Coloneqq
\tqbit\mid
\unittype\mid
\oc A\mid
A\limp B\mid
A\otimes B\mid
A\oplus B
\\
\begin{aligned}
\text{\emph{Term}}\;\;
M,N,L \Coloneqq{}&
x^A\mid
\tnew^A\mid
\tmeas^A\mid
U^A\mid
\tlam[n]{x^A}{M}\mid
MN\mid
\unitterm^n \mid
\tletin{\ttup{x^A,y^B}^n}{N}{M}
\\ \mid{}&
\ttup{M,N}^n \mid
\tinl[n][A,B]{M} \mid
\tinr[n][A,B]{N} \mid
\tmatchwith[n]{L}{x^A}{M}{y^B}{N}
\end{aligned}
\\
\text{\emph{Value}}\;\;
V,W \Coloneqq{}
x^A\mid
\tnew^A\mid
\tmeas^A\mid
U^A\mid
\unitterm^n \mid
\tlam[n]{x^A}{M}\mid
\ttup{V,W}^n \mid
\tinl[n][A,B]{V} \mid
\tinr[n][A,B]{W}
\end{gather*}
\caption{Types, terms and values of the quantum lambda calculus}
\label{tab:types-terms}
\vspace{-2ex}
\end{table}

The language consists of
\emph{types}, \emph{terms} and \emph{values} defined
in Table~\ref{tab:types-terms}.
We use obvious shorthand $\oc^n A=\oc\dotsb\oc A$
and $A^{\otimes n}=A\otimes\dotsb\otimes A$.
The \emph{subtyping relation} $\subtype$ on types
is defined by the rules
shown in Table~\ref{tab:typing-rules}(a).
In the definition of terms and values,
$n\in\N$ is a natural number; $x$ ranges over variables;
and $U$ ranges over $2^k\times 2^k$ unitary matrices for $k\ge 1$.
The (nullary) constructors $\tnew,\tmeas,U$ are called \emph{constants}
and sometimes referred to by $c$.
Clearly, values form a subclass of terms.
As usual, we identify terms up to $\alpha$-equivalence.

\begin{nremark}
\label{rem:indexedterms}
The terms are \emph{indexed terms} of~\cite{SelingerV2008},
which have explicit type annotations
(cf.\ Church-style vs. Curry-style in the simply-typed lambda calculus).
A typing derivation for an indexed term is unique in a suitable sense,
so that we can more easily obtain Lemma~\ref{lem:interpretation-well-defd}.
In fact, for the language of~\cite{SelingerV2008}
we can safely remove the type annotations~\cite[Corollary~1]{SelingerV2008}.
We conjecture that the same is true for our language,
which is left as a future work.
\end{nremark}

\begin{notation}
\label{notation:sumtype}
Following~\cite{SelingerV2009}
(and \cite{HasuoH2011,PaganiSV2014}), the language
has sum type $\oplus$ instead of the $\tbit$ type (which exists in~\cite{SelingerV2006}).
The $\tbit$ type and its constructors are emulated by
$\tbit\coloneqq\unittype\oplus\unittype$;
$\tzero^n\coloneqq\tinl[n][\unittype,\unittype]{\unitterm^n}$;
$\tone^n\coloneqq\tinr[n][\unittype,\unittype]{\unitterm^n}$;
and $\tifthen{L}{M}{N}\coloneqq\tmatchwith[0]{L}{x^{\unittype}}{M}{y^{\unittype}}{N}$,
with fresh variables $x,y$.
\end{notation}

The set $\FV(M)$ of \emph{free variables} is defined in the usual way.
A \emph{context} is a list $\Delta=x_1:A_1,\dotsc,x_n:A_n$
of variables $x_i$ and types $A_i$ where the variables $x_i$ are distinct.
We write $\abs{\Delta}=\{x_1,\dotsc,x_n\}$ and $\oc\Delta=x_1:\oc A_1,\dotsc,x_n:\oc A_n$.
We also write $\Delta|_M=\Delta\cap\FV(M)$ for the context
restricted to the free variables of $M$.

A \emph{typing judgement}, written as $\Delta\yields M:A$,
consists of a context $\Delta$, a term $M$ and a type $A$.
A typing judgement is \emph{valid} if it can be derived by
the \emph{typing rules} shown in Table~\ref{tab:typing-rules}(b).
In the rule \ruleref{ax2}, $c$ ranges over $\tnew$, $\tmeas$
and $2^k\times 2^k$ unitary matrices $U$;
and the types $A_c$ are defined as follows:
$A_{\tnew}=\tbit\limp\tqbit$,
$A_{\tmeas}=\tqbit\limp\oc\tbit$,
$A_U=\tqbit^{\otimes k}\limp\tqbit^{\otimes k}$.

\begin{table}[tb]
\newcommand{\MySkipAmount}{\vskip2.0ex plus.8ex minus.4ex}
\small
{\centering
  \AxiomC{\vphantom{$A_1$}}
  \UnaryInfC{$\oc^n\tqbit\subtype\oc^m\tqbit$}
  \DisplayProof
\quad
  \AxiomC{\vphantom{$A_1$}}
  \UnaryInfC{$\oc^n\unittype\subtype\oc^m\unittype$}
  \DisplayProof
\quad
  \AxiomC{$A_1\subtype B_1$}
  \AxiomC{$A_2\subtype B_2$}
  \BinaryInfC{$\oc^n(A_2\limp B_1)\subtype\oc^m(A_1\limp B_2)$}
  \DisplayProof
\MySkipAmount
  \AxiomC{$A_1\subtype B_1$}
  \AxiomC{$A_2\subtype B_2$}
  \BinaryInfC{$\oc^n(A_1\otimes A_2)\subtype\oc^m(B_1\otimes B_2)$}
  \DisplayProof
\quad
  \AxiomC{$A_1\subtype B_1$}
  \AxiomC{$A_2\subtype B_2$}
  \BinaryInfC{$\oc^n(A_1\oplus A_2)\subtype\oc^m(B_1\oplus B_2)$}
  \DisplayProof
\par}

\vspace{1ex}
\textbf{(a)} Rules for subtyping, with a condition
$(n=0\Rightarrow m=0)$ for each rule
\vspace{2ex}

{\centering
  \MyRightLabel{$\mathit{ex}$}{ex}
  \AxiomC{$\Delta,x:A,y:B,\Gamma\yields M:C$}
  \UnaryInfC{$\Delta,y:B,x:A,\Gamma\yields M:C$}
  \DisplayProof
\quad
  \MyRightLabel{$\mathit{ax_1}$}{ax1}
  \AxiomC{$A\subtype B$}
  \UnaryInfC{$\Delta,x:A\yields x^B:B$}
  \DisplayProof
\quad
  \MyRightLabel{$\mathit{ax_2}$}{ax2}
  \AxiomC{$\oc A_c\subtype B$}
  \UnaryInfC{$\Delta\yields c^B:B$}
  \DisplayProof
\MySkipAmount
  \MyRightLabel{$\limp.I_1$}{limp.I1}
  \AxiomC{$\Delta,x:A\yields M:B$}
  \UnaryInfC{$\Delta\yields \tlam[0]{x^A}{M}:A\limp B$}
  \DisplayProof
\quad
  \MyRightLabel{$\limp.I_2$}{limp.I2}
  \AxiomC{$\Gamma,\oc\Delta,x:A\yields M:B$}
  \AxiomC{$\FV(M)\cap\abs{\Gamma}=\emptyset$}
  \BinaryInfC{$\Gamma,\oc\Delta\yields \tlam[n+1]{x^A}{M}:\oc^{n+1}(A\limp B)$}
  \DisplayProof
\MySkipAmount
  \MyRightLabel{$\limp.E$}{limp.E}
  \AxiomC{$\oc\Delta,\Gamma_1\yields M:A\limp B$}
  \AxiomC{$\oc\Delta,\Gamma_2\yields N:A$}
  \BinaryInfC{$\oc\Delta,\Gamma_1,\Gamma_2\yields MN:B$}
  \DisplayProof
\MySkipAmount
  \MyRightLabel{$\top$}{top}
  \AxiomC{\vphantom{$\Gamma_1\oc^n A$}}
  \UnaryInfC{$\Delta\yields \unitterm^n:\oc^{n}\unittype$}
  \DisplayProof
\quad
  \MyRightLabel{$\otimes.I$}{tens.I}
  \AxiomC{$\oc\Delta,\Gamma_1\yields M:\oc^n A$}
  \AxiomC{$\oc\Delta,\Gamma_2\yields N:\oc^n B$}
  \BinaryInfC{$\oc\Delta,\Gamma_1,\Gamma_2\yields \ttup{M,N}^n:\oc^n(A\otimes B)$}
  \DisplayProof
\MySkipAmount
  \MyRightLabel{$\otimes.E$}{tens.E}
  \AxiomC{$\oc\Delta,\Gamma_1,x:\oc^n A,y:\oc^n B\yields M:C$}
  \AxiomC{$\oc\Delta,\Gamma_2\yields N:\oc^n (A\otimes B)$}
  \BinaryInfC{$\oc\Delta,\Gamma_1,\Gamma_2\yields
    \tletin{\ttup{x^A,y^B}^n}{N}{M}:C$}
  \DisplayProof
\MySkipAmount
  \MyRightLabel{$\oplus.I_1$}{sum.I1}
  \AxiomC{$\Delta\yields M:\oc^n A$}
  \UnaryInfC{$\Delta\yields\tinl[n][A,B]{M}:\oc^n(A\oplus B)$}
  \DisplayProof
\quad
  \MyRightLabel{$\oplus.I_2$}{sum.I2}
  \AxiomC{$\Delta\yields N:\oc^n B$}
  \UnaryInfC{$\Delta\yields\tinr[n][A,B]{N}:\oc^n(A\oplus B)$}
  \DisplayProof
\MySkipAmount
  \MyRightLabel{$\oplus.E$}{sum.E}
  \AxiomC{$\oc\Delta,\Gamma_1,x:\oc^n A\yields M:C$}
  \AxiomC{$\oc\Delta,\Gamma_1,y:\oc^n B\yields N:C$}
  \AxiomC{$\oc\Delta,\Gamma_2\yields L:\oc^n(A\oplus B)$}
  \TrinaryInfC{$\oc\Delta,\Gamma_1,\Gamma_2\yields
    \tmatchwith[n]{L}{x^A}{M}{y^B}{N}:C$}
  \DisplayProof
\par}

\vspace{1ex}
\textbf{(b)} Typing rules

\caption{Subtyping relation and typing rules}
\label{tab:typing-rules}
\vspace{-2ex}
\end{table}

The type system is affine (weak linear).
Each variable may occur at most once,
unless it has a duplicable type $\oc A$.
Substitution of the following form is admissible.

\begin{lemma}[Substitution]
\label{lem:subst}
If $\oc\Delta,\Gamma_1, x:A\yields M:B$
and $\oc\Delta,\Gamma_2\yields V:A$,
where $V$ is a value and
$\abs{\Gamma_1}\cap\abs{\Gamma_2}=\emptyset$,
then $\oc\Delta,\Gamma_1,\Gamma_2\yields \tsubst{M}{x}{V}:B$.
\qed
\end{lemma}

\noindent
Note, however, that we need to define
the substitution $\tsubst{M}{x}{V}$ with care.
For example, if $A\subtype A'$, $M=y^{A'\limp B}x^{A'}$
and $V=z^A$, then we substitute $z^{A'}$
(not $z^A$) for $x^{A'}$ in $M$.
See~\cite[\S2.5]{SelingerV2008}
or~\cite[\S9.1.4]{Valiron2008PhD}
for details.

\subsection{Operational Semantics}

The operational semantics is taken from~\cite{SelingerV2006,SelingerV2009},
but is adapted for indexed terms.

\begin{definition}
A \emph{quantum closure} is a triple $\qclos{\ket{\psi}}{\Psi}{M}$
with $m\in\N$ where:
\begin{itemize}
\item
$\ket{\psi}$ is a normalised vector of the Hilbert space $(\C^2)^{\otimes m}\cong\C^{2^m}$.
\item
$\Psi$ is a list of $m$ distinct variables,
written as $\ket{x_1\dotso x_m}$.
We write $\abs{\Psi}=\{x_1,\dotsc,x_m\}$,
and $\Psi(x_i)=i$ for the position of a variable in the list.
\item
$M$ is a term with $\FV(M)\subseteq\abs{\Psi}$.
\end{itemize}
We say a quantum closure $P=\qclos{\ket{\psi}}{\ket{x_1\dotso x_m}}{M}$
is \emph{well-typed of type $A$},
written as $P:A$,
if the typing judgement $x_1:\tqbit,\dotsc,x_m:\tqbit\yields M:A$ is valid.
We call $\qclos{\ket{\psi}}{\Psi}{V}$
a \emph{value closure} if $V$ is a value.
\end{definition}

\auxproof{We identify quantum closures up to the obvious $\alpha$-equivalence.}%

\begin{table}
\small
\setlength{\abovedisplayskip}{0pt}
\setlength{\belowdisplayskip}{0pt}
\setlength{\abovedisplayshortskip}{0pt}
\setlength{\belowdisplayshortskip}{0pt}
\begin{gather*}
\tag{$\limp$}\rulelabel{redu:beta.limp}
\qclos{\ket{\psi}}{\Psi}{(\tlam[0]{x^A}{M})V}
\redu[1]
\qclos{\ket{\psi}}{\Psi}{\tsubst{M}{x}{V}}
\\
\tag{$\otimes$}\rulelabel{redu:beta.tens}
\qclos{\ket{\psi}}{\Psi}{\tletin{\ttup{x^A,y^B}^n}{\ttup{V,W}^n}{M}}
\redu[1]
\qclos{\ket{\psi}}{\Psi}{\tsubsttwo{M}{x}{V}{y}{W}}
\\
\tag{$\oplus_1$}\rulelabel{redu:beta.sum1}
\qclos{\ket{\psi}}{\Psi}{\tmatchwith[n]{\tinl[n][A,B]{V}}{x^A}{M}{y^B}{N}}
\redu[1]
\qclos{\ket{\psi}}{\Psi}{\tsubst{M}{x}{V}}
\\
\tag{$\oplus_2$}\rulelabel{redu:beta.sum2}
\qclos{\ket{\psi}}{\Psi}{\tmatchwith[n]{\tinr[n][A,B]{W}}{x^A}{M}{y^B}{N}}
\redu[1]
\qclos{\ket{\psi}}{\Psi}{\tsubst{N}{y}{W}}
\end{gather*}

\vspace{1ex}
\textbf{(a)} Classical control

\begin{gather*}
\tag{$U$}\rulelabel{redu:U}
\qclos{\ket{\psi}}{\Psi}{U^{\tqbit^{\otimes k}\limp\tqbit^{\otimes k}}
\ttup{x_1^{\tqbit},\dotsc,x_k^{\tqbit}}^0}
\redu[1]
\qclos{\ket{\psi'}}{\Psi}{\ttup{x_1^{\tqbit},\dotsc,x_k^{\tqbit}}^0}
\\
\tag{$\tmeas_0$}\rulelabel{redu:meas0}
\qclos{\ket{\psi}}{\ket{x_1\dotso x_m}}{\tmeas^{\tqbit\limp\oc^n\tbit}
x_i^{\tqbit}}
\redu[p_0]
\qclos{\ket{\psi_0}}{\ket{x_1\dotso x_m}}{\tzero^n}
\\
\tag{$\tmeas_1$}\rulelabel{redu:meas1}
\qclos{\ket{\psi}}{\ket{x_1\dotso x_m}}{\tmeas^{\tqbit\limp\oc^n\tbit}
x_i^{\tqbit}}
\redu[p_1]
\qclos{\ket{\psi_1}}{\ket{x_1\dotso x_m}}{\tone^n}
\\
\tag{$\tnew_0$}\rulelabel{redu:new0}
\qclos{\ket{\psi}}{\ket{x_1\dotso x_m}}{\tnew^{A\limp\tqbit}\tilde{\tzero}}
\redu[1]
\qclos{\ket{\psi}\ket{0}}{\ket{x_1\dotso x_m y}}{y^{\tqbit}}
\\
\tag{$\tnew_1$}\rulelabel{redu:new1}
\qclos{\ket{\psi}}{\ket{x_1\dotso x_m}}{\tnew^{A\limp\tqbit}\tilde{\tone}}
\redu[1]
\qclos{\ket{\psi}\ket{1}}{\ket{x_1\dotso x_m y}}{y^{\tqbit}}
\end{gather*}

\vspace{1ex}
\textbf{(b)} Quantum data

\vspace{2ex}
\par If $\qclos{\ket{\psi}}{\Psi}{M}\redu[p]\qclos{\ket{\psi'}}{\Psi'}{M'}$,
the following are valid reductions (if well-formed).
\vspace{1ex}
\begin{gather*}
\qclos{\ket{\psi}}{\Psi}{MN}
    \redu[p]
    \qclos{\ket{\psi'}}{\Psi'}{M'N}
\qquad
\qclos{\ket{\psi}}{\Psi}{VM}
    \redu[p]
    \qclos{\ket{\psi'}}{\Psi'}{VM'}
\\
\qclos{\ket{\psi}}{\Psi}{\ttup{M,N}^n}
    \redu[p]
    \qclos{\ket{\psi'}}{\Psi'}{\ttup{M',N}^n}
\qquad
\qclos{\ket{\psi}}{\Psi}{\ttup{V,M}^n}
    \redu[p]
    \qclos{\ket{\psi'}}{\Psi'}{\ttup{V,M'}^n}
\\
\qclos{\ket{\psi}}{\Psi}{\tletin{\ttup{x^A,y^B}^n}{M}{N}}
    \redu[p]
    \qclos{\ket{\psi'}}{\Psi'}{\tletin{\ttup{x^A,y^B}^n}{M'}{N}}
\\
\qclos{\ket{\psi}}{\Psi}{\tinl[n][A,B]{M}}
    \redu[p]
    \qclos{\ket{\psi'}}{\Psi'}{\tinl[n][A,B]{M'}}
\\
\qclos{\ket{\psi}}{\Psi}{\tinr[n][A,B]{M}}
    \redu[p]
    \qclos{\ket{\psi'}}{\Psi'}{\tinr[n][A,B]{M'}}
\\
\qclos{\ket{\psi}}{\Psi}{\tmatchwith[n]{M}{x^A}{N}{y^B}{L}}
    \redu[p]
    \qclos{\ket{\psi'}}{\Psi'}{\tmatchwith[n]{M'}{x^A}{N}{y^B}{L}}
\end{gather*}

\vspace{1ex}
\textbf{(c)} Congruence rules
\caption{Reduction rules}
\label{tab:reduction-rules}
\vspace{-2ex}
\end{table}

\begin{definition}
A (small-step) \emph{reduction} $P\redu[p] Q$
consists of quantum closures $P,Q$ and $p\in[0,1]$,
meaning that $P$ reduces to $Q$ with probability $p$.
The valid reductions $P\redu[p] Q$ are given inductively by
the \emph{reduction rules} shown in Table~\ref{tab:reduction-rules}.
In the rules, $V$ and $W$ refer to values.
\auxproof{The rules in (a) and (b) are base cases respectively for classical
control and quantum data. The rules in (c) are inductive steps,
which reduce terms in a call-by-value strategy.}%
The `quantum data' rules (b)
correspond to the three basic operations
explained in \S\ref{sec:quantum-computation}.
In the rule \ruleref{redu:U},
$\ket{\psi'}$ is the state obtained by applying
the $2^k\times 2^k$ unitary matrix $U$
to the $k$ qubits of the position $\Psi(x_1),\dotsc,\Psi(x_k)$ in $\ket{\psi}$.
\auxproof{Specifically,
$\ket{\psi'}=\permop{\sigma}(U\otimes\Iop_{m-k}) \permop{\sigma}^{-1}\ket{\psi}$,
where $\permop{\sigma}\colon(\C^2)^{\otimes m}\to(\C^2)^{\otimes m}$
is an operator satisfying
$\permop{\sigma}\ket{\varphi_1}\dotsb\ket{\varphi_m}
=\ket{\varphi_{\sigma(1)}}\dotsb\ket{\varphi_{\sigma(m)}}$
for a permutation $\sigma$
of $\{1,\dotsc,m\}$ with $\sigma(i)=\Psi(x_i)$ for $i\le k$.\todo{introduce $\Iop$}}%
In the rule \ruleref{redu:meas0}, $p_0$ is the probability that
we obtain $0$ (`negative' in terms of \S\ref{sec:quantum-computation})
by measuring the $i$-th qubit of $\ket{\psi}$;
and $\ket{\psi_0}$ is the state after that.
The rule \ruleref{redu:meas1} is similar.
\auxproof{In the rules \ruleref{redu:meas0}
and \ruleref{redu:meas1},
for each $b\in\{0,1\}$,
$p_b$ is the probability that we observe $b$
by measuring the $i$-th qubit of $\ket{\psi}$ by basis $\{\ket{0},\ket{1}\}$;
$\ket{\psi_b}$ is the state after we observe $b$.
Specifically:
$p_b=
\braket{\psi}{P^i_b}{\psi}$ and
$\ket{\psi_b}
=P^i_b \ket{\psi}/\sqrt{p_b}$,
using the projection
$P^i_b=\Iop_{i-1}
\otimes\ket{b}\bra{b}\otimes
\Iop_{m-i}$.}%
In the rule \ruleref{redu:new0},
we denote by $\tilde{\tzero}$
any term of the form $\tinl[n][\oc^k\unittype,\oc^h\unittype]{\unitterm^{n+k}}$
(cf.\ Notation~\ref{notation:sumtype}).
The term $\tilde{\tone}$ in \ruleref{redu:new1} is similar.
\end{definition}

Reduction satisfies the following properties.

\begin{lemma}[Subject reduction]
If $P:A$ and $P\redu[p] Q$, then $Q:A$.
\qed
\end{lemma}
\auxproof{\begin{proof}
By induction on the derivation of the reduction $P\redu[p] Q$.
\end{proof}}%

\begin{lemma}[Progress]
\label{lem:progress}
Let $P:A$ be a well-typed quantum closure.
Then either $P$ is a value closure,
or there exists a quantum closure $Q$ such that $P\redu[p] Q$.
In the latter case,
there are at most two (up to $\alpha$-equivalence) single-step reductions from $P$,
and the total probability of all the single-step reductions from $P$ is $1$.
\qed
\end{lemma}

The next definitions follow~\cite{SelingerV2006,SelingerV2008a}.

\begin{definition}
We define the \emph{small-step reduction probability}
$\prob(P,Q)\in[0,1]$ for well-typed quantum closures
$P,Q$ by: $\prob(P,Q)=p$ if $P\to_p Q$;
$\prob(V,V)=1$ if $V$ is a value closure;
$\prob(P,Q)=0$ otherwise.
Lemma~\ref{lem:progress} guarantees that $\prob$
is a probabilistic system in a suitable sense.
For a well-typed quantum closure $P$
and a well-typed value closure $Z$,
the \emph{big-step reduction probability} $\Prob(P,Z)\in[0,1]$
is defined by
$\Prob(P,Z)=\lim_{n\to\infty}\prob^n(P,Z)$,
where
$\prob^1(P,Z)=\prob(P,Z)$
and $\prob^{n+1}(P,Z)=
\sum_{Q}\prob(P,Q)
\prob^{n}(Q,Z)$.
\end{definition}

\begin{definition}
For each $b\in\{\tzero^0,\tone^0\}$,
we define $\bigredu{P}{b}=\sum_{Z\in U_b}\Prob(P,Z)$,
where $U_b$ is the set of well-typed quantum closures
of the form $\qclos{\ket{\psi}}{\Psi}{b}$.
\end{definition}

We will use a strong normalisation result.
The proof is similar to~\cite[Lemma~33]{PaganiSV2014}.

\begin{lemma}[Strong normalisation]
\label{lem:str-norm}
Let $P:A$ be a well-typed quantum closure.
Then there is no infinite sequence of reductions
$P\redu[p_1] P_1\redu[p_2]P_2\redu[p_3]\dotsb$.
\end{lemma}
\begin{proof}[Proof (Sketch)]
Clearly it suffices to prove
the strong normalisation for
the underlying (non-deterministic) reductions $M\redu N$ on
terms. We add a constant $c^{\tqbit}$ to replace
free variables $x^{\tqbit}$.
We then define a translation $\transl{(-)}$
from the quantum lambda calculus (with $c^{\tqbit}$)
to a simply-typed lambda calculus
with product, unit, sum types and
constants $\tnew,\tmeas,U,c^{\tqbit}$.
The translation $\transl{(-)}$ forgets
the $\bang$ modality, and translates the let constructor via
$\transl{(\tletin{\ttup{x,y}}{N}{M})}=
(\tlam{z}{(\tlam{x}{\tlam{y}}{\transl{M}}})
\tfst{z}\tsnd{z})\transl{N}$.
We can prove the strong normalisation for
the simply-typed lambda calculus via standard techniques.
\end{proof}

\section{Denotational Semantics Using von Neumann Algebras}
\label{sec:denotational-semantics}

\subsection{Facts about von Neumann Algebras}

We need the following notation and facts
concerning von Neumann algebras.
Those facts for which we could not find 
proof
in the literature will be discussed in  the next section.

Let $(\vNAMIU,\otimes,\C)$ be the symmetric monoidal category (SMC) of von Neumann algebras
and normal MIU-maps~\cite[Prop.~7.2]{Kornell2012},
and $(\vNACPsU,\otimes,\C)$ the SMC of von Neumann algebras and
normal CPsU-maps
(where~$\otimes$ is the spatial tensor product)~\cite{cho2014}.
Note that the unit $\C$ is initial in $\vNAMIU$ (but not in $\vNACPsU$).
Both categories have products given by direct sums $\oplus$
(with the supremum norm~\cite[Def.~3.4]{Takesaki1}).
To interpret the quantum lambda calculus,
we will use the following pair of (lax) symmetric monoidal adjunctions,
\begin{equation}
\label{eq:monoidal-adjunctions}
\xymatrix@C=3pc{
(\Set^{\op},\times,1)
\ar@/_2ex/[r]_-{\linf}
\ar@{}[r]|-{\bot}
&
\ar@/_2ex/[l]_-{\nsp}
(\vNAMIU,\otimes,\C)
\ar@/_2ex/@{^{ (}->}[r]_{\incfun}
\ar@{}[r]|{\bot}
&
(\vNACPsU,\otimes,\C)
\ar@/_2ex/[l]_{\qcomp}
}
\end{equation}
where $\Set^{\op}$ is the opposite of the category
$\Set$ of sets and functions, considered as a SMC
via cartesian products (i.e.\ coproducts in $\Set^{\op}$).
The functor $\incfun$ is the inclusion functor;
the other functors are explained in the next section.
Note that $\incfun$ is strict symmetric monoidal
and strictly preserves products.
The following facts are important:
\begin{itemize}
\item
$\vNAMIU$ is a
co-closed SMC~\cite{Kornell2012}.
This means the endofunctor $(-)\otimes\scrA$ on $\vNAMIU$
has a left adjoint $\frexp{\scrA}{(-)}$.
The von Neumann algebra $\frexp{\scrA}{\scrB}$ is called
the \emph{free exponential} in~\cite{Kornell2012}.
\item
The counit of the adjunction $\nsp\dashv \linf$ is an isomorphism
(see Corollary~\ref{cor:first-adjunction-unit-iso}).
\item
The functors $\nsp,\linf$
and the adjunction $\nsp\dashv \linf$ are \emph{strong} monoidal
(see Corollary~\ref{cor:first-adjunction-monoidal}).
\item
Moreover, the functors $\nsp$ and $\linf$ preserves products
(see Cor.~\ref{cor:sp-preserves-coproducts-and-tensors} and
Lem.~\ref{lem:ell-tensor}).
\end{itemize}

The tensor product $\otimes$ distributes over products $\oplus$ in $\vNAMIU$,
as $\scrA\otimes(\scrB\oplus \scrC)\cong
(\scrA\otimes\scrB)\oplus(\scrA\otimes \scrC)$,
since $\scrA\otimes(-)$ is a right adjoint
and thus preserves products.
We denote the canonical isomorphism
by $\theta_{\scrA,\scrB,\scrC}\colon(\scrA\otimes\scrB)\oplus(\scrA\otimes \scrC)\to\scrA\otimes(\scrB\oplus \scrC)$.

\begin{wrapfigure}{r}{.45\columnwidth}
\centering
\vspace{-.5em}
\alwaysDoubleLine
\AxiomC{$
\scrB\overset{f}{\longto} \scrC\otimes\scrA
$ in $\vNACPsU$}
\UnaryInfC{$
\qcomp\scrB\longto \scrC\otimes\scrA
$ in $\vNAMIU$}
\UnaryInfC{$
\frexp{\scrA}{(\qcomp\scrB)}=
\scrA\limp\scrB\overset{g}{\longto} \scrC
$ in $\vNAMIU$}
\DisplayProof
\vspace{-.5em}
\end{wrapfigure}
We define a `Kleisli co-exponential' $\limp$ by
$\scrA\limp\scrB\coloneqq\frexp{\scrA}{(\qcomp\scrB)}$.
We have the bijective correspondence
as shown on the right.
We write $\Lambda f=g$ for the 
MIU-map $\scrA\limp\scrB\to \scrC$ corresponding to $f$.
We also write
$\varepsilon_{\scrA,\scrB}=\Lambda^{-1}\id\colon\scrB\to (\scrA\limp\scrB)\otimes\scrA$
for the co-evaluation map, i.e.\ the CPsU-map corresponding to
$\id\colon \scrA\limp\scrB\to\scrA\limp\scrB$.
Then $(\Lambda f\otimes\id)\circ\varepsilon=f$
by the naturality of the bijective correspondence.

We write $\lem=\linf\circ \nsp$ for the strong symmetric
monoidal monad on $\vNAMIU$ induced by
the left-hand adjunction of \eqref{eq:monoidal-adjunctions}.
The unit and multiplication are denoted by $\eta$ and $\mu$ respectively.
From the fact that the counit of $\nsp\dashv \linf$ is an isomorphism,
it easily follows that $\lem$ is an idempotent monad, i.e.\ the multiplication
$\mu\colon \lem^2\To\lem$ is an isomorphism.
Note also that $\lem$ preserves products.
We denote the structure isomorphisms by:
$d^{\lem}_{\C}\colon\C\to\lem\C$;
$d^{\lem}_{\scrA,\scrB}\colon\lem\scrA\otimes\lem\scrB\to
\lem(\scrA\otimes\scrB)$; and
$e^{\lem}_{\scrA,\scrB}\colon\lem\scrA\oplus\lem\scrB
\to\lem(\scrA\oplus\scrB)$.

Because the adjunction $\nsp\dashv \linf$ satisfies a dual condition
to a linear-non-linear model~\cite{Benton1995} (see also~\cite{Mellies2002,Schalk2004}),
the monad $\lem$ has a property which is dual to a \emph{linear exponential comonad}.
Thus each object of the form $\lem \scrA$ is equipped with
a map $\monmap_{\scrA}\colon\lem\scrA\otimes\lem\scrA\to\lem\scrA$
which, with a unique map $\invbang_{\lem\scrA}\colon\C \to\lem\scrA$,
makes $\lem \scrA$ into a $\otimes$-monoid in $\vNAMIU$.

\begin{nremark}
\label{rem:concrete-model-of-qlc}
One can summarise these facts by saying that
the opposite $\vNAMIU^{\op}$ is a
\emph{(weak) linear category for duplication}~\cite{SelingerV2008,SelingerV2009};
and moreover $\vNAMIU^{\op}$ is
a \emph{concrete model of the quantum lambda calculus}
defined by Selinger and Valiron~\cite[\S1.6.8]{SelingerV2009}.
Although they gave the definition of concrete models of the quantum lambda calculus,
results on them (e.g.\ how to interpret the quantum lambda calculus;
adequacy of models) have never been given.
In the remainder of the section, therefore, we will give
the interpretation of the language in von Neumann algebras
concretely, and then prove its adequacy.
\end{nremark}

\subsection{The Interpretation of Types and Typing Judgements}

We interpret types as von Neumann algebras,
i.e.\ objects in $\vNAMIU$ / $\vNACPsU$, as follows.
\begin{align*}
\sem{\tqbit}&=\Mat_2
&
\sem{\unittype}&=\C
&
\sem{\oc A}&=\lem\sem{A}
\\
\sem{A\limp B}&=\sem{A}\limp\sem{B}
&
\sem{A\otimes B}&=\sem{A}\otimes\sem{B}
&
\sem{A\oplus B}&=\sem{A}\oplus\sem{B}
\end{align*}

\begin{remark}
One  familiar
with  Fock space
might be surprised to realise that
$\sem{\oc \tqbit} = \{0\}$,
because
there is no normal MIU-map $\varphi\colon \Mat_2\to\mathbb{C}$.
The intuition here may be 
that no part of a qubit can be duplicated,
and so the assumption of a duplicable qubit
amounts to nothing.
This is also the interpretation
of $\oc\tqbit$ intended by
Selinger and Valiron, see~\cite[\S5]{SelingerV2008}.
\end{remark}
\begin{remark}
The interpretation
of a function type~$A\limp B$
is obtained by abstract means,
and at this point we know very little about
it. (Might it be
as intangible as an ultrafilter?)
However,
applying~$\oc$ makes the function type almost trivial: 
after \S4,
it will be clear that
\begin{equation*}
\sem{\oc(A\limp B)} 
\ =\  \linf(\{\, f\colon \sem{B} \stackrel{\mathrm{CPsU}}{
	\longrightarrow}\sem{A} \,\}).
\end{equation*}
\end{remark}

The interpretation of the subtyping relation $A\subtype B$
is a `canonical' map $\sem{B}\to\sem{A}$ in $\vNAMIU$,
which exists uniquely by a coherence property for
an idempotent (co)monad; see~\cite[\S8.3.2]{Valiron2008PhD} for details.
For instance, we have $\sem{A\limp\oc B\subtype \oc A\limp \oc\oc B}=
\eta_{\sem{A}}\limp\mu_{\sem{B}}$.

Contexts $\Delta=x_1:A_1,\dotsc,x_n:A_n$
are interpreted as $\sem{\Delta}=\sem{A_1}\otimes\dotsb\otimes\sem{A_n}$.
We shall treat the monoidal structure $(\otimes,\C)$
as if it were strict monoidal,
which is justified by the coherence theorem for monoidal categories.

The interpretations
$\sem{\tnew}$, $\sem{\tmeas}$ and $\sem{U}$ of constants
are defined using the maps $f_{\tnew}\colon \Mat_2\to \C^2$,
$f_{\tmeas}\colon\C^2\to \Mat_2$
and $f_{U}\colon\Mat_2^{\otimes k}
\to \Mat_2^{\otimes k}$
given in \S\ref{sec:quantum-computation},
as follows.
\begin{align*}
\sem{\tnew}&=
\eta_{\C}^{-1}
\circ\lem\Lambda f_{\tnew}
\colon
\sem{\oc A_{\tnew}}
=\lem(\C^2\limp \Mat_2)
\longto \C
\\
\sem{\tmeas}&=
\eta_{\C}^{-1}\circ
\lem\Lambda (f_{\tmeas}\circ \eta_{\C^2}^{-1})
\colon
\sem{\oc A_{\tmeas}}=
\lem(\Mat_2\limp\lem\C^2) \longto \C
\\
\sem{U}&=
\eta_{\C}^{-1}
\circ
\lem\Lambda f_U
\colon
\sem{\oc A_U}=
\lem(\Mat_2^{\otimes k}\limp\Mat_2^{\otimes k}) \longto \C
\end{align*}

We now give the interpretation $\sem{\Delta\yields M:A}$ of a typing judgement
as a map $\sem{A}\to\sem{\Delta}$ in $\vNACPsU$.
The definition is similar to~\cite{HasuoH2014}.
First we define a normal CPsU-map
$\sem{\Delta\yields M:A}^{\FV}\colon \sem{A}\to\sem{\Delta|_M}$
(recall that $\Delta|_M=\Delta\cap\FV(M)$)
by induction on the derivation of the typing judgement
as shown in Table~\ref{tab:denotation-typing-judgements}.
We then define $\sem{\Delta\yields M:A}
\coloneqq
(\sem{A}
\xrightarrow{\sem{\Delta\yields M:A}^{\FV}}
\sem{\Delta|_M}
\xrightarrow{\iota}
\sem{\Delta})$.
Here and in Table~\ref{tab:denotation-typing-judgements},
we use the following notations
(often suppressing subscripts).
Let $\gamma_{\scrA,\scrB}\colon\scrA\otimes\scrB\to\scrB\otimes\scrA$
denote the symmetry isomorphism.
For contexts $\Delta\subseteq \Gamma$,
we write $\iota\colon\sem{\Delta}\to\sem{\Gamma}$ for the `injection' map
defined via unique MIU-maps $\invbang_{\scrA}\colon\C\to\scrA$.
\auxproof{
For instance:
\[
\iota=
\bigl(\sem{A}\otimes\sem{C}\cong\sem{A}\otimes\C\otimes\sem{C}
\xrightarrow{\id\otimes\invbang\otimes\id}
\sem{A}\otimes\sem{B}\otimes\sem{C}\bigr)
\]}%
For a context $\oc\Delta,\Gamma_1,\Gamma_2$,
we define the map
$\mergmap\colon\sem{\oc\Delta,\Gamma_1}\otimes\sem{\oc\Delta,\Gamma_2}\to\sem{\oc\Delta,\Gamma_1,\Gamma_2}$
via monoid structures $\monmap_{\sem{A}}\colon \sem{\oc A}\otimes\sem{\oc A}\to\sem{\oc A}$
and symmetry maps $\gamma$.
The map $d^{\lem}_{\Delta}\colon\sem{\oc\Delta}\to\lem\sem{\Delta}$
can be defined using
$d^{\lem}_{\scrA,\scrB}\colon\lem \scrA\otimes\lem \scrB\to\lem(\scrA\otimes \scrB)$.
We write $\mu_{\Delta}\colon\sem{\oc\oc\Delta}\to\sem{\oc\Delta}$
for $\mu_{\sem{A_1}}\otimes\dotsb\otimes\mu_{\sem{A_n}}$;
$d^{\lem^n}\colon\lem^n \scrA\otimes\lem^n \scrB\to\lem^n(\scrA\otimes \scrB)$
for $\lem^{n-1}d^{\lem}\circ\dotsb\circ d^{\lem}$;
and $e^{\lem^n}\colon\lem^n \scrA\oplus\lem^n \scrB\to\lem^n(\scrA\oplus \scrB)$
similarly.
Projection maps and tupling for direct sums, products in $\vNACPsU$,
are denoted by $\pi_i\colon \scrA_1\oplus\scrA_2\to \scrA_i$
and $\tup{f,g}\colon \scrA\to \scrB\oplus\scrC$.

\auxproof{Note that $d^{\lem}_{\C}=\eta_{\C}\colon\C\to\lem\C$}%

\begin{table}
\newcommand{\MySkipAmount}{\vskip2.0ex plus.8ex minus.4ex}
\small\centering
  \AxiomC{$
    \sem{\Delta,x:A,y:B,\Gamma\yields M:C}^{\FV}
    =
    \sem{C}
    \xrightarrow{f}
    \sem{(\Delta,x:A,y:B,\Gamma)|_M}
    $}
  \UnaryInfC{$
    \sem{\Delta,y:B,x:A,\Gamma\yields M:C}^{\FV}
    =
    (\id_{\sem{\Delta|_M}}\otimes\gamma\otimes\id_{\sem{\Gamma|_M}})\circ f
    \enspace\text{(if $x,y\in\FV(M)$)};\enspace
    f\enspace\text{(otherwise)}$}
  \DisplayProof
\MySkipAmount
  $\sem{\Delta,x:A\yields x^B:B}^{\FV}
  =
  \sem{B}
  \xrightarrow{\sem{A\subtype B}}
  \sem{A}$
\qquad
  $\sem{\Delta\yields c^B:B}^{\FV}
  =
  \sem{B}
  \xrightarrow{\sem{\oc A_c\subtype B}}
  \sem{\oc A_c}
  \xrightarrow{\sem{c}}
  \C$
\MySkipAmount
  $\sem{\Delta\yields \unitterm^n:\oc^n \unittype}^{\FV}
  =
  \sem{\oc^n\top}
  \xrightarrow{\sem{\oc\top\subtype\oc^n\top}}
  \lem\C
  \xrightarrow{(d^{\lem}_{\C})^{-1}}
  \C$
\MySkipAmount
  \AxiomC{$
    \sem{\Delta,x:A\yields M:B}^{\FV}=
    \sem{B}
    \xrightarrow{f}
    \sem{(\Delta,x:A)|_M}
    $}
  \UnaryInfC{$
    \sem{\Delta\yields \tlam[0]{x^A}{M}:A\limp B}^{\FV}
    =
    \sem{A}\limp\sem{B}
    \xrightarrow{\Lambda f'}
    \sem{\Delta|_M}
    $}
  \DisplayProof
\quad where:
\smash{$\vcenter{\xymatrix@C=.7pc@R=.5pc{
\sem{B}\ar[r]^-{f}
\ar[dr]_-{f'}&
\sem{(\Delta,x:A)|_M}\ar[d]^{\iota}\\
&
\sem{\Delta|_M}\otimes \sem{A}\\
}}$}
\MySkipAmount
  \AxiomC{$
    \sem{\Gamma,\oc\Delta,x:A\yields M:B}^{\FV}=
    \sem{B}\xrightarrow{f}
    \sem{(\Gamma,\oc\Delta,x:A)|_M}
    $}
  \UnaryInfC{$
    \begin{aligned}
      &\sem{\Gamma,\oc\Delta\yields \tlam[n+1]{x^A}{M}:A\limp B}^{\FV}
      =\sem{\oc^{n+1}(A\limp B)}
      \xrightarrow{\sem{\oc(A\limp B)\subtype\oc^{n+1}(A\limp B)}}
      \lem(\sem{A}\limp\sem{B})
      \\
      &\qquad
      \xrightarrow{\lem(\Lambda f')}
      \lem\sem{\oc\Delta|_M}
      \xrightarrow{(d^{\lem})^{-1}}
      \sem{\oc\oc\Delta|_M}
      \xrightarrow{\mu}
      \sem{\oc\Delta|_M}
      =\sem{(\oc\Delta,\Gamma)|_M}
\quad\text{($f'$ defined similarly)}
    \end{aligned}
    $}
  \DisplayProof
\MySkipAmount
  \AxiomC{$
    \begin{aligned}
      \sem{\oc\Delta,\Gamma_1
        \yields M:A\limp B}^{\FV}
      &=\sem{A\limp B}
      \xrightarrow{f}
      \sem{(\oc\Delta,\Gamma_1)|_M}
      \\
      \sem{\oc\Delta,\Gamma_2
        \yields N:A}^{\FV}
      &=\sem{A}
      \xrightarrow{g}
      \sem{(\oc\Delta,\Gamma_2)|_N}
    \end{aligned}
    $}
  \UnaryInfC{$
    \begin{aligned}
      &\sem{\oc\Delta,\Gamma_1,\Gamma_2
        \yields MN:B}^{\FV}
      =
      \sem{B}
      \xrightarrow{\varepsilon}
      \sem{A\limp B}\otimes\sem{A}
      \xrightarrow{f\otimes g}
      \sem{(\oc\Delta,\Gamma_1)|_M}
      \otimes
      \sem{(\oc\Delta,\Gamma_2)|_N}
      \\
      &\qquad\qquad
      \xrightarrow{\iota\otimes\iota}
      \sem{(\oc\Delta,\Gamma_1)|_{MN}}
      \otimes
      \sem{(\oc\Delta,\Gamma_2)|_{MN}}
      \xrightarrow{\splmap}
      \sem{(\oc\Delta,\Gamma_1,\Gamma_2)|_{MN}}
    \end{aligned}
    $}
  \DisplayProof
\MySkipAmount
  \AxiomC{$
    \begin{aligned}
      \sem{\oc\Delta,\Gamma_1
        \yields M:\oc^n A}^{\FV}
      =\sem{\oc^n A}
      \xrightarrow{f}
      \sem{(\oc\Delta,\Gamma_1)|_M}
      \qquad
      \sem{\oc\Delta,\Gamma_2
        \yields N:\oc^n B}^{\FV}
      =\sem{\oc^n B}
      \xrightarrow{g}
      \sem{(\oc\Delta,\Gamma_2)|_N}
    \end{aligned}
    $}
  \UnaryInfC{$
    \begin{aligned}
      &\sem{\oc\Delta,\Gamma_1,\Gamma_2
        \yields \ttup{M,N}^n:\oc^n (A\otimes B)}^{\FV}
      =
      \sem{\oc^n (A\otimes B)}
      \xrightarrow{(d^{\lem^n})^{-1}}
      \sem{\oc^n A}\otimes\sem{\oc^n B}
      \xrightarrow{f\otimes g}
      \\
      &\quad
      \sem{(\oc\Delta,\Gamma_1)|_M}
      \otimes
      \sem{(\oc\Delta,\Gamma_2)|_N}
      \xrightarrow{\iota\otimes\iota}
      \sem{(\oc\Delta,\Gamma_1)|_{\ttup{M,N}}}
      \otimes
      \sem{(\oc\Delta,\Gamma_2)|_{\ttup{M,N}}}
      \xrightarrow{\splmap}
      \sem{(\oc\Delta,\Gamma_1,\Gamma_2)|_{\ttup{M,N}}}
    \end{aligned}
    $}
  \DisplayProof
\MySkipAmount
  \AxiomC{$
    \begin{aligned}
      \sem{\oc\Delta,\Gamma_1,
        x:\oc^n A,y:\oc^n B
        \yields M:C}^{\FV}
      &=\sem{C}
      \xrightarrow{f}
      \sem{(\oc\Delta,\Gamma_1,x:\oc^n A,
        y:\oc^n B)|_M}
      \\
      \sem{\oc\Delta,\Gamma_2
        \yields N:\oc^n (A\otimes B)}^{\FV}
      &=\sem{\oc^n (A\otimes B)}
      \xrightarrow{g}
      \sem{(\oc\Delta,\Gamma_2)|_M}
    \end{aligned}
    $}
  \UnaryInfC{$
    \begin{aligned}
      &\sem{\oc\Delta,\Gamma_1,\Gamma_2
        \yields \tletin{\ttup{x^A,y^B}^n}{N}{M}:C
      }^{\FV}
      =
      \sem{C}
      \xrightarrow{f}
      \sem{(\oc\Delta,\Gamma_1,x:\oc^n A,
        y:\oc^n B)|_M}
      \\
      &\quad
      \xrightarrow{\iota}
      \sem{(\oc\Delta,\Gamma_1)|_{\mathtt{let}\dotso}}
      \otimes
      \sem{\oc^n A}\otimes \sem{\oc^n B}
      \xrightarrow{\id\otimes d^{\lem^n}}
      \sem{(\oc\Delta,\Gamma_1)|_{\mathtt{let}\dotso}}
      \otimes
      \sem{\oc^n (A\otimes B)}
      \xrightarrow{\id\otimes g}
      \\
      &\quad
      \sem{(\oc\Delta,\Gamma_1)|_{\mathtt{let}\dotso}}
      \otimes
      \sem{(\oc\Delta,\Gamma_2)|_N}
      \xrightarrow{\id\otimes\iota}
      \sem{(\oc\Delta,\Gamma_1)|_{\mathtt{let}\dotso}}
      \otimes
      \sem{(\oc\Delta,\Gamma_2)|_{\mathtt{let}\dotso}}
      \xrightarrow{\splmap}
      \sem{(\oc\Delta,\Gamma_1,\Gamma_2)|_{\mathtt{let}\dotso}}
    \end{aligned}
    $}
  \DisplayProof
\MySkipAmount
  \AxiomC{$
    \sem{\Delta\yields M:\oc^n A}^{\FV}=
    \sem{\oc^n A}
    \xrightarrow{f}
    \sem{\Delta|_M}
    $}
  \UnaryInfC{$
    \sem{\Delta\yields
      \tinl[n][A,B]{M}:\oc^n(A\oplus B)}^{\FV}
    =
    \sem{\oc^n(A\oplus B)}
    \xrightarrow{\lem^n\pi_1}
    \sem{\oc^n A}
    \xrightarrow{f}
    \sem{\Delta|_M}
    $}
  \DisplayProof
\MySkipAmount
  \AxiomC{$
    \sem{\Delta\yields N:\oc^n B}^{\FV}=
    \sem{\oc^n B}
    \xrightarrow{g}
    \sem{\Delta|_N}
    $}
  \UnaryInfC{$
    \sem{\Delta\yields
      \tinr[n][A,B]{N}:\oc^n(A\oplus B)}^{\FV}
    =
    \sem{\oc^n(A\oplus B)}
    \xrightarrow{\lem^n\pi_2}
    \sem{\oc^n B}
    \xrightarrow{g}
    \sem{\Delta|_N}
    $}
  \DisplayProof
\MySkipAmount
  \AxiomC{$
    \begin{aligned}
      \sem{\oc\Delta,\Gamma_1,
        x:\oc^n A
        \yields M:C}^{\FV}
      &=\sem{C}
      \xrightarrow{f}
      \sem{(\oc\Delta,\Gamma_1,x:\oc^n A)|_M}
      \\
      \sem{\oc\Delta,\Gamma_1,
        y:\oc^n B
        \yields N:C}^{\FV}
      &=\sem{C}
      \xrightarrow{g}
      \sem{(\oc\Delta,\Gamma_1,y:\oc^n B)|_N}
      \\
      \sem{\oc\Delta,\Gamma_2
        \yields L:\oc^n (A\oplus B)}^{\FV}
      &=\sem{\oc^n (A\oplus B)}
      \xrightarrow{h}
      \sem{(\oc\Delta,\Gamma_2)|_L}
    \end{aligned}
    $}
  \UnaryInfC{$
    \begin{aligned}
      &\sem{\oc\Delta,\Gamma_1,\Gamma_2
        \yields\tmatchwith[n]{L}{x^A}{M}{y^B}{N}:C}^{\FV}
      =
      \\
      &
      \sem{C}
      \xrightarrow{\tup{f,g}}
      \sem{(\oc\Delta,\Gamma_1,x:\oc^n A)|_M}
      \oplus
      \sem{(\oc\Delta,\Gamma_1,y:\oc^n B)|_N}
      \xrightarrow{\iota\oplus\iota}
      \\
      &
      (\sem{(\oc\Delta,\Gamma_1)|_{\mathtt{match}\dotso}}
      \otimes\sem{\oc^n A})
      \oplus
      (\sem{(\oc\Delta,\Gamma_1)|_{\mathtt{match}\dotso}}
      \otimes\sem{\oc^n B})
      \xrightarrow{\theta}
      \sem{(\oc\Delta,\Gamma_1)|_{\mathtt{match}\dotso}}
      \otimes
      (\sem{\oc^n A}\oplus \sem{\oc^n B})
      \\
      &
      \xrightarrow{\id\otimes e^{\lem^n}}
      \sem{(\oc\Delta,\Gamma_1)|_{\mathtt{match}\dotso}}
      \otimes\sem{\oc^n (A\oplus B)}
      \xrightarrow{\id\otimes h}
      \sem{(\oc\Delta,\Gamma_1)|_{\mathtt{match}\dotso}}\otimes
      \sem{(\oc\Delta,\Gamma_2)|_L}
      \\
      &
      \xrightarrow{\id\otimes\iota}
      \sem{(\oc\Delta,\Gamma_1)|_{\mathtt{match}\dotso}}\otimes
      \sem{(\oc\Delta,\Gamma_2)|_{\mathtt{match}\dotso}}
      \xrightarrow{\splmap}
      \sem{(\oc\Delta,\Gamma_1,\Gamma_2)|_{\mathtt{match}\dotso}}
    \end{aligned}
    $}
  \DisplayProof
\MySkipAmount
\caption{Inductive definition of the interpretation of typing judgements}
\label{tab:denotation-typing-judgements}
\end{table}

Note that the interpretation $\sem{\Delta\yields M:A}$ is
defined by induction on typing derivations.
Because we use indexed terms,
it is not hard to prove the following fact
by induction on a typing derivation $\Pi$.

\begin{lemma}
\label{lem:interpretation-well-defd}
Suppose that $\Delta\yields M:A$ is valid
with a derivation $\Pi$, and
so is $\Delta'\yields M:A$ with $\Pi'$.
Then $\sem{\Pi'}^{\FV}=
\sigma\circ\sem{\Pi}^{\FV}$,
where $\sigma\colon\sem{\Delta|_M}
\to\sem{\Delta'|_M}$ is a (unique by coherence) isomorphism
that permutes $\Delta|_M$ to $\Delta'|_M$.
In particular, $\sem{\Delta\yields M:A}$ is well-defined,
not depending on derivations.
\qed
\end{lemma}

Let $\qclos{\ket{\psi}}{\ket{x_1\dotsc x_n}}{M}:A$
be a well-typed quantum closure.
The mapping $A\mapsto \braket{\psi}{A}{\psi}$
defines a normal CPU-map $\braket{\psi}{-}{\psi}\colon\Mat_2^{\otimes m}\to\C$.
The interpretation of the quantum closure is defined by:
\[
\sem{\qclos{\ket{\psi}}{\ket{x_1\dotsc x_n}}{M}:A}
\;\coloneqq\;
\sem{A}
\xrightarrow{\sem{x_1:\tqbit,\dotsc,x_n:\tqbit\yields M:A}}
\Mat_2^{\otimes n}
\xrightarrow{\braket{\psi}{-}{\psi}}
\C
\]

\subsection{Adequacy of the Denotational Semantics}

The next soundness/invariance
for the small-step reduction
is a key result to obtain adequacy.
Note that for normal CPsU-maps $f_1,\dotsc,f_n\colon \scrA\to\scrB$
and $r_i\in[0,1]$ with $\sum_i r_i\le 1$,
the (convex) sum $\sum_i r_if_i$ of maps
is defined in the obvious pointwise manner and is a normal CPsU-map.

\begin{proposition}[Soundness for the small-step reduction]
\label{prop:soundness-small-step}
Let $P:A$ be a well-typed quantum closure.
Then $\sem{P:A}=\sum_{Q}\prob(P,Q)\sem{Q:A}$.
\end{proposition}
\begin{proof}
See Appendix~\ref{sec:soundness-small-step}.
\end{proof}

\begin{Auxproof}
\begin{lemma}
\label{soundness-n-small-step}
Let $P:A$ be a well-typed quantum closure.
Then for all $n\ge 1$,
\[
\sem{P:A}=\sum_{Q}\prob^n(P,Q)\sem{Q:A}
\enspace.
\]
\end{lemma}
\begin{proof}
By induction on $n$.
The base case is Proposition~\ref{prop:soundness-small-step}.
For the induction step,
\begin{align*}
\sem{P:A}&=\sum_{Q}\prob(P,Q)\sem{Q:A}
\\
&=\sum_{Q}\prob(P,Q)
\Bigl(\sum_{R}\prob^n(Q,R)\sem{R:A}\Bigr)
&&\text{by IH}
\\
&=\sum_{Q}\sum_{R}\prob(P,Q)\prob^n(Q,R)\sem{R:A}
\\
&=\sum_{R}
\Bigl(\sum_{Q}\prob(P,Q)\prob^n(Q,R)\Bigr)
\sem{R:A}
\\
&=\sum_{R}\prob^{n+1}(P,R)\sem{R:A}
\end{align*}
\end{proof}
\end{Auxproof}

\begin{proposition}[Soundness for the big-step reduction]
\label{prop:soundness-big-step}
Let $P:A$ be a well-typed quantum closure.
Then $\sem{P:A}=\sum_{Z}\Prob(P,Z)\sem{Z:A}$,
where $Z$ runs over well-typed value closures.
\end{proposition}
\begin{proof}
By Lemmas~\ref{lem:progress} and~\ref{lem:str-norm},
$\Prob(P,Z)
\overset{\mathrm{def}}{=}
\lim_{n\to\infty}\prob^n(P,Z)
=
\prob^m(P,Z)$
for some $m$.
It is then easy to obtain
$\sem{P:A}=\sum_{Q}\prob^m(P,Q)\sem{Q:A}$
by induction on $m$, using Proposition~\ref{prop:soundness-small-step}.
\end{proof}

\begin{theorem}[Adequacy]
Let $P:\tbit$ be a quantum closure of type $\tbit$.
For the interpretation $\sem{P:\tbit}\colon \C\oplus\C\to\C$,
we have $\bigredu{P}{\tzero}=\sem{P:\tbit}(1,0)$ and
$\bigredu{P}{\tone}=\sem{P:\tbit}(0,1)$.
\end{theorem}
\begin{proof}
By Proposition~\ref{prop:soundness-big-step} we have
$\sem{P:\tbit}=\sum_{Z}\Prob(P,Z)\sem{Z:\tbit}$.
Note that for each well-typed value closure
$\qclos{\ket{\psi}}{\Psi}{V}:\tbit$,
either $V=\tzero^0$ or $V=\tone^0$.
Then the assertion follows since
$\sem{\qclos{\ket{\psi}}{\Psi}{\tzero^0}:\tbit}(\lambda,\rho)=\lambda$
and $\sem{\qclos{\ket{\psi}}{\Psi}{\tone^0}:\tbit}(\lambda,\rho)=\rho$.
\end{proof}

\section{Technical Results about von Neumann Algebras}
\label{technical-result-vNA}

Let us sketch how we obtained
the two monoidal adjunctions in~\eqref{eq:monoidal-adjunctions}.
\begin{definition}
Let~$\linf(X)$
denote the von Neumann algebra
of bounded maps $f\colon X\to\mathbb{C}$
on a set~$X$.
Addition, multiplication, involution, suprema,  and so on, are computed
coordinatewise in~$\linf(X)$.
In fact, $\linf(X)$ is simply the $X$-fold
product in~$\vNAMIU$ of~$\mathbb{C}$
with $\varphi\mapsto \varphi(x)$
as~$x$-th projection.
We extend $X\mapsto \linf(X)$ to a functor
$\ell^\infty\colon \Set^\op \to \vNAMIU$
by defining
$\linf(f)(\varphi) = 
\varphi\circ f$
for every map $f\colon X\to Y$
(in $\Set$) and $\varphi\in \linf(Y)$.

Let $\nsp(\mathscr{A})$
be the `normal spectrum' of a von Neumann algebra~$\mathscr{A}$,
i.e.\ the set of normal MIU-maps $\varphi\colon \mathscr{A}\to\mathbb{C}$.
We extend $\mathscr{A}\mapsto\nsp(\mathscr{A})$
to a functor
$\nsp\colon \vNAMIU \to \Set^\op$
by defining $\nsp(f)(\varphi)=\varphi\circ f$
for every normal MIU-map $f\colon \mathscr{A}\to\mathscr{B}$
and~$\varphi\in\nsp(\mathscr{B})$ (it is simply a hom-functor $\vNAMIU(-,\C)$).
\end{definition}

Note that any normal MIU-map $f\colon \mathscr{A}\to\ell^\infty(X)$
gives a map $g\colon X\to \nsp(\mathscr{A})$
by ``swapping arguments''
 ---
$g(x)(\varphi)=f(\varphi)(x)$
---
and with a little bit more  work, we get:
\begin{lemma}
There is an adjunction~$\nsp \dashv \linf$.\qed
\end{lemma}

The following two lemmas
describe the normal spectrum 
of direct products and tensors
of von Neumann algebras,
and can be proven 
using standard techniques.
\begin{lemma}
Let~$I$ be a set,
and for each~$i\in I$, let~$\mathscr{A}_i$ be a von Neumann algebra.
For each~$\omega\in \nsp(\bigoplus_{i\in I}\mathscr{A}_i)$,
there is~$i\in I$ 
and~$\tilde\omega \in \nsp(\mathscr{A}_i)$
with $\omega = \tilde\omega\circ \pi_{i}$.\qed
\end{lemma}

\begin{lemma}
Let~$\mathscr{A}_1$ and~$\mathscr{A}_2$
be von Neumann algebras.
Then for every~$\omega\in \nsp(\mathscr{A}_1\otimes \mathscr{A}_2)$
there are unique~$\omega_1\in\mathscr{A}_1$
and~$\omega_2\in\mathscr{A}_2$
with $\omega (a_1\otimes a_2) = \omega_1(a_1)\,\cdot\,\omega_2(a_2)$
for all~$a_i\in\mathscr{A}_i$.\qed
\end{lemma}
\begin{corollary}
\label{cor:sp-preserves-coproducts-and-tensors}
The functor~$\nsp\colon \vNAMIU \to \Set^\op$
preserves  products, and tensors.\qed
\end{corollary}
Using that $\linf(X)$
is the $X$-fold product of~$\mathbb{C}$
in~$\vNAMIU$ we get:
\begin{corollary}
\label{cor:first-adjunction-unit-iso}
The counit 
of the adjunction~$\nsp\dashv\linf$ is an isomorphism.\qed
\end{corollary}

\begin{lemma}
\label{lem:ell-tensor}
Let~$X$ and~$Y$ be sets.
There is a normal MIU-isomorphism
\begin{equation*}
\varphi\colon 
 \ell^\infty(X)\otimes \ell^\infty(Y)
\longrightarrow
\ell^\infty(X\times Y)
\qquad\text{given by}\qquad
\varphi(f\otimes g)(x,y) \ =\  f(x)\cdot g(y).
\end{equation*}
\end{lemma}
\begin{proof}
Use the proof of Proposition~9.2 from~\cite{cho2014}.
\end{proof}

\begin{corollary}
\label{cor:first-adjunction-monoidal}
The adjunction~$\nsp \dashv \linf$ 
is strong monoidal.\qed
\end{corollary}
Let us turn to the second adjunction in~\eqref{eq:monoidal-adjunctions}.
In~\cite{bram2014} 
it is shown how 
the following result follows
from Freyd's Adjoint Functor Theorem
(see Theorem V.6.2 of~\cite{maclane1978}).
\begin{theorem}[\cite{bram2014}]
\label{thm:second-adjoint}
The inclusion~$\mathcal{J}\colon \vNAMIU
\to \vNACPsU$
has a left adjoint.\qed
\end{theorem}

\begin{corollary}
\label{cor:second-adjoint}
The category~$\vNACPsU$ is isomorphic
to the co-Kleisli category
of the comonad~$\mathcal{F}\circ \mathcal{J}$
on~$\vNAMIU$
induced by~$\mathcal{F}\dashv \mathcal{J}$.\qed
\end{corollary}
\begin{proof}
See Theorem~9 of~\cite{bram2014},
or do Exercise VI.5.2 of~\cite{maclane1978}
(and use the fact
that an equivalence of categories 
which is bijective on objects
is an isomorphism).
\end{proof}

\begin{corollary}
\label{cor:second-adjoint-monoidal}
The adjunction 
$\mathcal{F}\dashv\mathcal{J}$ is symmetric monoidal.
\end{corollary}
\begin{proof}
Clearly, $\mathcal{J}\colon \vNAMIU
\to \vNACPsU$
is strict symmetric monoidal.
From this fact alone,
it follows that the adjunction
$\mathcal{F}\dashv\mathcal{J}$
is symmetric monoidal, see Prop.~14 of~\cite{Mellies2009}.
\end{proof}

In our model of the quantum lambda calculus
the von Neumann algebras of the form $\ell^\infty(X)$
serve as the interpretation of the \emph{duplicable types}
(of the form~$!A$),
because $\ell^\infty(X)$ carries a $\otimes$-monoid structure.
Among all von Neumann algebras~$\linf(X)$
is arguably quite special
and  one might wonder if
there is a broader class of von Neumann algebras
that might serve as the interpretation of duplicable types
(such as the class of all commutative von Neumann algebras,
which includes~$L^\infty[0,1]$).
The following result
settles this matter: no.
Due to space constraints,
the proof will appear somewhere else.

\begin{theorem}
\label{thm:duplicable}
For a von Neumann algebra~$\mathscr{A}$
the following are equivalent.
\begin{enumerate}
\item
There is a \emph{duplicator} on~$\mathscr{A}$,
that is, 
a normal positive unital map $\mu\colon \mathscr{A}\otimes\mathscr{A}
\to\mathscr{A}$
such that $\mu(1\otimes a) = a = \mu(a\otimes 1)$
and
$\mu(a \otimes  \mu(b\otimes c)) = \mu(\mu(a\otimes b)\otimes c)$
for all~$a,b,c\in\mathscr{A}$.
\label{thm:duplicable-1}
\item
\label{thm:duplicable-2}
$\mathscr{A}$ is isomorphic to~$\ell^\infty(X)$
for some set~$X$.
\end{enumerate}
Moreover,
there is at most one duplicator on~$\mathscr{A}$.
\qed
\end{theorem}

\begin{corollary}
$\linf(\nsp(\mathscr{A}))$
is the free $\otimes$-monoid
on~$\mathscr{A}$ from~$\vNAMIU$.\qed
\end{corollary}

\section{Final Remarks}
\label{sec:final-remarks}

We have given a  rather concrete proof of adequacy
for the sake of clarity.
However, it seems that we only used the fact that~$\vNAMIU$
is a `concrete model of the quantum lambda calculus'
(see Remark~\ref{rem:concrete-model-of-qlc}),
and that~$\vNACPsU$ is `suitably' enriched over convex sets.
Thus an abstract result
might be distilled
from our work
stating that any concrete model of the quantum lambda calculus
is adequate
when suitably enriched over convex sets,
but we have not pursued this.

We believe selling points of our model
are that
it is a model for Selinger and Valiron's original
quantum lambda calculus~\cite{SelingerV2006}
(in Selinger and Valiron's linear fragment~\cite{SelingerV2008a}
the $\bang$ modality is absent;
in Hasuo and Hoshino's
language~\cite{HasuoH2011}
the tensor type $\tqbit\otimes\tqbit$
does not represent two qubits;
and
only function types may be duplicable, $\oc(A\limp B)$,
in the language of Pagani et al.~\cite{PaganiSV2014});
that it is adequate
(Malherbe's model~\cite{Malherbe2010,MalherbeSS2013}
is not known to be);
that the interpretation of~$!$ is rather simple;
and that it is formed using von Neumann algebras,
a mathematical classic.

We believe our model could be improved
by a more concrete description of~$\sem{A\limp B}$
(as all the other models have),
and by features
such as recursion and inductive types
(present in e.g.~Hasuo and Hoshino's
and Pagani's models),
which leaves us with ample
material for future research.

\subparagraph*{Acknowledgements}

We thank Chris Heunen for spotting a typo.


\begin{thebibliography}{10}

\bibitem{Benton1995}
P.~N. Benton.
\newblock A mixed linear and non-linear logic: Proofs, terms and models.
\newblock In {\em CSL '94}, volume 933 of {\em LNCS}, pages 121--135. Springer,
  1995.

\bibitem{birkhoff1936}
Garrett Birkhoff and John von Neumann.
\newblock The logic of quantum mechanics.
\newblock {\em Annals of mathematics}, pages 823--843, 1936.

\bibitem{born1925}
Max Born and Pascual Jordan.
\newblock Zur quantenmechanik.
\newblock {\em Zeitschrift f{\"u}r Physik}, 34(1):858--888, 1925.

\bibitem{cho2014}
Kenta Cho.
\newblock Semantics for a quantum programming language by operator algebras.
\newblock In {\em QPL 2014}, volume 172 of {\em EPTCS}, pages 165--190, 2014.

\bibitem{cho2015}
Kenta Cho, Bart Jacobs, Bas Westerbaan, and Abraham Westerbaan.
\newblock Quotient-comprehension chains.
\newblock In {\em QPL 2015}, volume 195 of {\em EPTCS}, pages 136--147, 2015.

\bibitem{dauns1972}
John Dauns.
\newblock Categorical $w^*$-tensor product.
\newblock {\em Transactions of the American Mathematical Society},
  166:439--456, 1972.

\bibitem{dauns1978}
John Dauns.
\newblock {Enveloping $W^*$-algebras}.
\newblock {\em Journal of Mathematics}, 8(4), 1978.

\bibitem{foulis1994}
David~J. Foulis and Mary~K. Bennett.
\newblock Effect algebras and unsharp quantum logics.
\newblock {\em Foundations of Physics}, 24(10):1331--1352, 1994.

\bibitem{guichardet1966}
Alain Guichardet.
\newblock Sur la cat{\'e}gorie des algebres de von {Neumann}.
\newblock {\em Bull. Sci. Math}, 90(2):41--64, 1966.

\bibitem{HasuoH2011}
Ichiro Hasuo and Naohiko Hoshino.
\newblock Semantics of higher-order quantum computation via geometry of
  interaction.
\newblock In {\em {LICS} 2011}, pages 237--246. IEEE, 2011.

\bibitem{HasuoH2014}
Ichiro Hasuo and Naohiko Hoshino.
\newblock Semantics of higher-order quantum computation via geometry of
  interaction.
\newblock Extended version of \cite{HasuoH2011}, preprint, 2014.

\bibitem{heisenberg1925}
Werner Heisenberg.
\newblock Quantum-theoretical re-interpretation of kinematic and mechanical
  relations.
\newblock {\em Z. Phys}, 33:879--893, 1925.

\bibitem{jordan1933}
Pascual Jordan.
\newblock {\em {\"U}ber verallgemeinerungsm{\"o}glichkeiten des formalismus der
  quantenmechanik}.
\newblock Weidmann, 1933.

\bibitem{kadison1951}
Richard~V. Kadison.
\newblock {\em A representation theory for commutative topological algebra}.
\newblock American Mathematical Society, 1951.

\bibitem{KadisonRingrose2}
Richard~V. Kadison and John~R. Ringrose.
\newblock {\em Fundamentals of the theory of operator algebras: Advanced
  theory}, volume~2.
\newblock American Mathematical Society, 1997.

\bibitem{kadison1997}
Richard~V. Kadison and John~R. Ringrose.
\newblock {\em Fundamentals of the Theory of Operator Algebras: Elementary
  Theory}, volume~1.
\newblock American Mathematical Society, 1997.

\bibitem{kaplansky1951}
Irving Kaplansky.
\newblock Projections in {B}anach algebras.
\newblock {\em Annals of Mathematics}, pages 235--249, 1951.

\bibitem{Kornell2012}
Andre Kornell.
\newblock Quantum collections.
\newblock {\em arXiv preprint arXiv:1202.2994v1}, 2012.

\bibitem{lorenz1969}
Falko Lorenz.
\newblock Die epimorphismen der ringe von operatoren.
\newblock {\em Archiv der Mathematik}, 20(1):48--53, 1969.

\bibitem{maclane1978}
Saunders Mac~Lane.
\newblock {\em Categories for the working mathematician}.
\newblock Springer, second edition, 1998.

\bibitem{Malherbe2010}
Octavio Malherbe.
\newblock {\em Categorical models of computation: partially traced categories
  and presheaf models of quantum computation}.
\newblock PhD thesis, University of Ottawa, 2010.

\bibitem{MalherbeSS2013}
Octavio Malherbe, Philip Scott, and Peter Selinger.
\newblock Presheaf models of quantum computation: An outline.
\newblock In {\em Computation, Logic, Games, and Quantum Foundations. {T}he
  Many Facets of Samson Abramsky}, volume 7860 of {\em LNCS}, pages 178--194.
  Springer, 2013.

\bibitem{Mellies2002}
Paul-Andr{\'e} Melli{\`e}s.
\newblock Categorical models of linear logic revisited.
\newblock Available on the author's website, 2002.

\bibitem{Mellies2009}
Paul-Andr{\'e} Melli{\`e}s.
\newblock Categorical semantics of linear logic.
\newblock {\em Panoramas et synth{\`e}ses}, 27:1--196, 2009.

\bibitem{murray1936}
Francis~J. Murray and John von Neumann.
\newblock On rings of operators.
\newblock {\em Annals of Mathematics}, pages 116--229, 1936.

\bibitem{nielsen2010}
Michael~A Nielsen and Isaac~L Chuang.
\newblock {\em Quantum computation and quantum information}.
\newblock Cambridge university press, 2010.

\bibitem{PaganiSV2014}
Michele Pagani, Peter Selinger, and Beno{\^\i}t Valiron.
\newblock Applying quantitative semantics to higher-order quantum computing.
\newblock In {\em {POPL} 2014}, pages 647--658. ACM, 2014.

\bibitem{paschke1973}
William~L Paschke.
\newblock {Inner product modules over $B^*$-algebras}.
\newblock {\em Transactions of the American Mathematical Society},
  182:443--468, 1973.

\bibitem{paulsen2002}
Vern Paulsen.
\newblock {\em Completely bounded maps and operator algebras}, volume~78.
\newblock Cambridge University Press, 2002.

\bibitem{rennela2014}
Mathys Rennela.
\newblock {Towards a quantum domain theory: Order-enrichment and fixpoints in
  $W^*$-algebras}.
\newblock {\em Electronic Notes in Theoretical Computer Science}, 308:289--307,
  2014.

\bibitem{ruan1988}
Zhong-Jin Ruan.
\newblock {Subspaces of $C^*$-algebras}.
\newblock {\em Journal of functional analysis}, 76(1):217--230, 1988.

\bibitem{Sakai2012}
Sh{\^o}ichir{\^o} Sakai.
\newblock {\em {$C^*$-algebras and $W^*$-algebras}}.
\newblock Springer, 2012.

\bibitem{Schalk2004}
Andrea Schalk.
\newblock What is a categorical model of linear logic?
\newblock Available on the author's website, 2004.

\bibitem{segal1947}
Irving~E. Segal.
\newblock Postulates for general quantum mechanics.
\newblock {\em Annals of Mathematics}, pages 930--948, 1947.

\bibitem{selinger2004}
Peter Selinger.
\newblock Towards a quantum programming language.
\newblock {\em Mathematical Structures in Computer Science}, 14(04):527--586,
  2004.

\bibitem{SelingerV2005}
Peter Selinger and Beno{\^\i}t Valiron.
\newblock A lambda calculus for quantum computation with classical control.
\newblock In {\em {TLCA} 2005}, volume 3461 of {\em LNCS}, pages 354--368.
  Springer, 2005.

\bibitem{SelingerV2006}
Peter Selinger and Beno{\^\i}t Valiron.
\newblock A lambda calculus for quantum computation with classical control.
\newblock {\em Mathematical Structures in Computer Science}, 16(3):527--552,
  2006.

\bibitem{SelingerV2008}
Peter Selinger and Beno{\^\i}t Valiron.
\newblock A linear-non-linear model for a computational call-by-value lambda
  calculus (extended abstract).
\newblock In {\em {FoSSaCS} 2008}, volume 4962 of {\em LNCS}, pages 81--96.
  Springer, 2008.

\bibitem{SelingerV2008a}
Peter Selinger and Beno{\^\i}t Valiron.
\newblock On a fully abstract model for a quantum linear functional language:
  (extended abstract).
\newblock In {\em {QPL} 2006}, volume 210 of {\em ENTCS}, pages 123--137, 2008.

\bibitem{SelingerV2009}
Peter Selinger and Beno{\^\i}t Valiron.
\newblock Quantum lambda calculus.
\newblock In {\em Semantic Techniques in Quantum Computation}, pages 135--172.
  Cambridge University Press, 2009.

\bibitem{Takesaki1}
Masamichi Takesaki.
\newblock {\em Theory of operator algebras {I}}.
\newblock Springer, 2002.

\bibitem{Valiron2008PhD}
Beno{\^\i}t Valiron.
\newblock {\em Semantics for a Higher Order Functional Programming Language for
  Quantum Computation}.
\newblock PhD thesis, University of Ottawa, 2008.

\bibitem{vantonder2004}
Andr{\'e} van Tonder.
\newblock A lambda calculus for quantum computation.
\newblock {\em SIAM Journal on Computing}, 33(5):1109--1135, 2004.

\bibitem{neumann1927}
John von Neumann.
\newblock Wahrscheinlichkeitstheoretischer aufbau der quantenmechanik.
\newblock {\em Nachrichten von der Gesellschaft der Wissenschaften zu
  G{\"o}ttingen, Mathematisch-Physikalische Klasse}, 1927:245--272, 1927.

\bibitem{bram2014}
Abraham Westerbaan.
\newblock Quantum programs as {Kleisli} maps.
\newblock {\em arXiv preprint arXiv:1501.01020v2}, 2014.

\end{thebibliography}

\appendix

\section{Proof of Soundness for the Small-Step Reduction}
\label{sec:soundness-small-step}

\auxproof{
\begin{lemma}
Let $\Delta\yields V:A$
be a valid typing judgement with $V$ a value.
Then $\sem{\Delta\yields V:A}^{\FV}$ (and hence
$\sem{\Delta\yields V:A}$) is a MIU-map.
\end{lemma}
\begin{proof}
By induction on values $V$.
Note that the interpretation
of each constant $\sem{c}\colon\sem{\oc A_c}\to\C$ is by definition
MIU.
\end{proof}
}%

We need some results on the denotational semantics.

\begin{lemma}
\label{lem:denot-subst}
Suppose that $\oc\Delta,\Gamma_1,x:A\yields M:B$
and $\oc\Delta,\Gamma_2\yields V:A$,
so that $\oc\Delta,\Gamma_1,\Gamma_2\yields \tsubst{M}{x}{V}:B$
by Lemma~\ref{lem:subst}.
Then the following diagram commute.
\[
\xymatrix@C+5pc@R-1pc{
\sem{B} \ar[r]^-{\sem{\oc\Delta,\Gamma_1,\Gamma_2\yields \tsubst{M}{x}{V}:B}}
\ar[d]_{\sem{\oc\Delta,\Gamma_2,x:A\yields M:B}}&
\sem{\oc\Delta,\Gamma_1,\Gamma_2}
\\
\sem{\oc\Delta,\Gamma_1}\otimes\sem{A}
\ar[r]^-{\id\otimes\sem{\oc\Delta,\Gamma_2\yields V:A}} &
\sem{\oc\Delta,\Gamma_1}\otimes
\sem{\oc\Delta,\Gamma_2}
\ar[u]_{\splmap}
}
\]
\end{lemma}
\begin{proof}
By induction on $M$.
Note that
the interpretation of a value is MIU.
\end{proof}

\begin{lemma}
\label{lem:denot-beta-redu}
We have the following equations,
when terms $M,N$ and values $V,W$ are appropriately well-typed.
\begin{align*}
\sem{(\tlam[0]{x^A}{M})V}
&= \sem{\tsubst{M}{x}{V}}
\\
\sem{\tletin{\ttup{x^A,y^B}^n}{\ttup{V,W}^n}{M}}
&= \sem{\tsubsttwo{M}{x}{V}{y}{W}}
\\
\sem{\tmatchwith[n]{\tinl[n][A,B]{V}}{x^A}{M}{y^B}{N}}
&= \sem{\tsubst{M}{x}{V}}
\\
\sem{\tmatchwith[n]{\tinr[n][A,B]{W}}{x^A}{M}{y^B}{N}}
&= \sem{\tsubst{N}{y}{W}}
\end{align*}
Here we abbreviate $\sem{\Delta\yields M:A}$ to $\sem{M}$.
\end{lemma}
\begin{proof}
Straightforward, using Lemma~\ref{lem:denot-subst}.
\end{proof}

To prove Proposition~\ref{prop:soundness-small-step} by induction,
we need to strengthen the statement into Lemma~\ref{lem:soundness-strongthened}.
Note that $\sem{P:A}^{(0)}=\sem{P:A}$.

\begin{definition}
Let $\qclos{\ket{\psi}}{\ket{x_1\dotsc x_m}}{M}:A$
be a well-typed quantum closure
such that $x_i\notin\FV(M)$ for all $i\le l$.
Then we define:
\[
\sem{\qclos{\ket{\psi}}{\ket{x_1\dotsc x_m}}{M}:A}^{(l)}
=
\Mat_2^{\otimes l}
\otimes \sem{A}
\xrightarrow{\id\otimes
\sem{x_{l+1}:\tqbit,\dotsc,x_m:\tqbit\yields M:A}}
\Mat_2^{\otimes m}
\xrightarrow{\braket{\psi}{-}{\psi}}
\C
\]
\end{definition}

\begin{lemma}
\label{lem:soundness-strongthened}
Let $P=\qclos{\ket{\psi}}{\ket{x_1\dotsc x_m}}{M}:A$
be a well-typed quantum closure
such that $x_i\notin\FV(M)$ for all $i\le l$.
Then
$\sem{P:A}^{(l)}=\sum_{Q}\prob(P,Q)\sem{Q:A}^{(l)}$.
\end{lemma}
\begin{proof}
We prove it by induction on terms $M$.
If $M$ is a value, then it holds by the definition of $\prob$.
In the other induction steps, we prove the assertion by cases.

Consider the induction step for $MN$,
and the case where $M$ is not a value.
Then the only possible reductions
from $P=\qclos{\psi}{\Psi}{MN}$
are $\qclos{\psi}{\Psi}{MN}\redu[p]\qclos{\psi'}{\Psi'}{M'N}$
when $\qclos{\psi}{\Psi}{M}\redu[p]\qclos{\psi'}{\Psi'}{M'}$.
Without loss of generality,\footnote{%
A permutation of variables in $\Psi$ which keeps the first $l$ variables,
with the permutation of the corresponding qubits in $\ket{\psi}$, does not change
$\sem{P:A}^{(l)}$. The same is true for the operational
semantics~\cite[\S3.2]{SelingerV2008a}.} we may assume that
\[
l=\ket{y_1\dotso y_l z_1\dotso z_k x_1\dotso x_h}
\]
such that $x_1:\tqbit,\dotsc,x_h:\tqbit\yields M:A\limp B$
and $z_1:\tqbit,\dotsc,z_k:\tqbit\yields N:A$.
We will simply write $\sem{M}$ for
$\sem{x_1:\tqbit,\dotsc,x_h:\tqbit\yields M:A\limp B}$
and $\sem{N}$ similarly.
Then
\begin{align*}
&\sem{\qclos{\psi}{\Psi}{MN}:B}^{(l)}
\\
&=\Mat_2^{\otimes l}\otimes \sem{B}
\xrightarrow{\id\otimes\varepsilon}
\Mat_2^{\otimes l}\otimes \sem{A\limp B}\otimes \sem{A}
\xrightarrow{\id\otimes\sem{M}\otimes\sem{N}}
\Mat_2^{\otimes l}\otimes \Mat_2^{\otimes h}
\otimes \Mat_2^{\otimes k}
\\
&\qquad\xrightarrow{\id\otimes\gamma}
\Mat_2^{\otimes l}\otimes \Mat_2^{\otimes k}
\otimes \Mat_2^{\otimes h}
\xrightarrow{\braket{\psi}{-}{\psi}}
\C
\\
&=\Mat_2^{\otimes l}\otimes \sem{B}
\xrightarrow{\id\otimes\varepsilon}
\Mat_2^{\otimes l}\otimes \sem{A\limp B}\otimes \sem{A}
\xrightarrow{\id\otimes\gamma}
\Mat_2^{\otimes l}\otimes\sem{A}\otimes\sem{A\limp B}
\\
&\qquad\xrightarrow{\id\otimes\sem{N}\otimes\id}
\Mat_2^{\otimes (l+k)}\otimes\sem{A\limp B}
\xrightarrow{\id\otimes\sem{M}}
\Mat_2^{\otimes (l+k)}
\otimes \Mat_2^{\otimes h}
\xrightarrow{\braket{\psi}{-}{\psi}}
\C
\end{align*}
Let $\qclos{\psi}{\Psi}{M}\redu[p_i]
\qclos{\psi_i}{\Psi_i}{M_i}$ ($i\in I$)
be all the reductions from $\qclos{\psi}{\Psi}{M}$.
By IH, we have
\begin{align*}
&\Mat_2^{\otimes(l+k)}\otimes\sem{A\limp B}
\xrightarrow{\id\otimes\sem{M}}
\Mat_2^{\otimes(l+k)}
\otimes \Mat_2^{\otimes h}
\xrightarrow{\braket{\psi}{-}{\psi}}
\C
\\
&=\sem{\qclos{\psi}{\Psi}{M}:A\limp B}^{(l+k)}
\\
&=
\sum_{i\in I} p_i\sem{\qclos{\psi_i}{\Psi_i}{M_i}:A\limp B}^{(l+k)}
\\
&=\sum_{i\in I} p_i\Bigl(
\Mat_2^{\otimes(l+k)}\otimes\sem{A\limp B}
\xrightarrow{\id\otimes\sem{M_i}}
\Mat_2^{\otimes(l+k)}
\otimes \Mat_2^{\otimes h_i}
\xrightarrow{\braket{\psi_i}{-}{\psi_i}}
\C
\Bigr)
\end{align*}
It is then straightforward to see that
$\sem{\qclos{\psi}{\Psi}{MN}:B}^{(l)}=\sum_i p_i\sem{\qclos{\psi_i}{\Psi_i}{M_iN}:B}^{(l)}$.

Next consider the case where $M=U$
and $N=\ttup{x_1,\dotsc,x_k}^0$.
Without loss of generality we may assume $P=\qclos{\ket{\psi}}{\Psi}{U\ttup{x_1,\dotsc,x_k}}$ with
$l=\ket{y_1\dotso y_l x_1\dotso x_k z_1\dotso z_h}$.
The only reduction from $P$ is
$\qclos{\ket{\psi}}{\Psi}{U\ttup{x_1,\dotsc,x_k}}\redu[1]
\qclos{\ket{\psi'}}{\Psi}{\ttup{x_1,\dotsc,x_k}}
\eqqcolon Q$,
where $\ket{\psi'}=(\Iop_l\otimes U\otimes\Iop_h)\ket{\psi}$
($\Iop_n$ denotes the $2^n\times 2^n$ identity matrix).
We need to show that $\sem{P:\tqbit^{\otimes k}}^{(l)}=\sem{Q:\tqbit^{\otimes k}}^{(l)}$.
Note that
\[
\sem{x_1:\tqbit,\dotsc,x_k:\tqbit\yields U\ttup{x_1,\dotsc,x_k}:\tqbit^{\otimes k}}
=f_U\colon\Mat_2^{\otimes k}\to\Mat_2^{\otimes k}
\]
Thus we have
\[
\sem{P:\tqbit^{\otimes k}}^{(l)}=
\Mat_2^{\otimes l}\otimes
\Mat_2^{\otimes k}
\xrightarrow{\id\otimes f_U}
\Mat_2^{\otimes l}\otimes
\Mat_2^{\otimes k}
\xrightarrow{\id\otimes \iota}
\Mat_2^{\otimes (l+k+h)}
\xrightarrow{\braket{\psi}{-}{\psi}}
\C
\]
On the other hand, we have
\[
\sem{x_1:\tqbit,\dotsc,x_k:\tqbit\yields \ttup{x_1,\dotsc,x_k}:\tqbit^{\otimes k}}
=\id\colon\Mat_2^{\otimes k}\to\Mat_2^{\otimes k}
\]
and hence
\[
\sem{Q:\tqbit^{\otimes k}}^{(l)}=
\Mat_2^{\otimes l}\otimes
\Mat_2^{\otimes k}\xrightarrow{\id\otimes\iota}
\Mat_2^{\otimes (l+k+h)}
\xrightarrow{\braket{\psi'}{-}{\psi'}}
\C
\]
For each elementary tensor
$A\otimes B\in\Mat_2^{\otimes l}\otimes\Mat_2^{\otimes k}$,
\begin{align*}
\sem{P:\tqbit^{\otimes k}}^{(l)}(A\otimes B)
&=\braket{\psi}{(\id\otimes\iota)((\id\otimes f_U)(A\otimes B))}{\psi}
\\
&=\braket{\psi}{A\otimes (U^{\dagger}BU)\otimes\Iop_h}{\psi}
\\
&=\braket{\psi}{
(\Iop_l\otimes U^{\dagger}\otimes \Iop_h)
(A\otimes B\otimes\Iop_h)(\Iop_l\otimes U\otimes \Iop_h)}{\psi}
\\
&=\braket{\psi'}{(\id\otimes\iota)(A\otimes B)}{\psi'}
\\
&=\sem{Q:\tqbit^{\otimes k}}^{(l)}(A\otimes B)
\end{align*}
We conclude that $\sem{P:\tqbit^{\otimes k}}^{(l)}=\sem{Q:\tqbit^{\otimes k}}^{(l)}$.

Consider the case where $MN$ is of the form $(\tlam{x}{M})V$.
Only the reduction is $\qclos{\ket{\psi}}{\Psi}{(\tlam{x}{M})V}
\redu[1]\qclos{\ket{\psi}}{\Psi}{\tsubst{M}{x}{V}}$.
The assertion holds immediately by Lemma~\ref{lem:denot-beta-redu}.

The other cases in the induction step $MN$ can be shown similarly.
We can prove the other induction steps similarly by cases.
\end{proof}

\end{document}